\documentclass[a4paper,11pt]{article}
 
\usepackage{authblk}
\usepackage{fullpage}
\usepackage{graphicx,url}
\usepackage[utf8]{inputenc}
\usepackage[english]{babel}
\usepackage{amsmath,amsthm}
\usepackage{algorithm}
\usepackage[noend]{algpseudocode}
\usepackage{enumitem}
 
\usepackage{pgfplots}
\pgfplotsset{compat=1.18}
 
\newtheorem{theorem}{Theorem}

\newtheorem{corollary}[theorem]{Corollary}
\newtheorem{lemma}[theorem]{Lemma}

\DeclareMathOperator{\nim}{nim}
\DeclareMathOperator{\mex}{mex}
\def\xor{\oplus}
 
\newcommand{\nimpc}[1]{\operatorname{nim}_{\ell_k}(P''_{#1})}
\newcommand{\nimpb}[1]{\operatorname{nim}_{\ell_k}(P'_{#1})}
\newcommand{\nimp}[1]{\operatorname{nim}_{\ell_k}(P_{#1})}
\newcommand{\nimptc}[1]{\operatorname{nim}_{\ell_2}(P''_{#1})}
\newcommand{\nimptb}[1]{\operatorname{nim}_{\ell_2}(P'_{#1})}
\newcommand{\nimpt}[1]{\operatorname{nim}_{\ell_2}(P_{#1})}
 
\DeclareMathOperator{\I}{I}
\DeclareMathOperator{\hull}{hull}
\DeclareMathOperator{\CIG}{CIG}
\DeclareMathOperator{\CHG}{CHG}
\DeclareMathOperator{\gc}{g}
\DeclareMathOperator{\mc}{m}

\title{Hull Games of Induced Path Convexities in Graphs}
 
\author[1]{Eurinardo Costa}
\author[1]{Leon Almeida}
\author[2]{Rudini Sampaio}
 
\affil[1]{Campus Russas, Universidade Federal do Ceará, Russas, CE, Brazil}
\affil[2]{Dept. Computação, Universidade Federal do Ceará, Fortaleza, CE, Brazil}

\date{}
\begin{document}
 
\maketitle
 
\begin{abstract}
In 1984, Frank Harary introduced the first convexity games in graphs, all of them based on the geodesic convexity, which is the graph convexity related to shortest paths. In 2024, Ara\'ujo et al. obtained the first PSPACE-hardness proofs on some of these geodesic games and generalized them to any graph convexity. In this paper, we investigate convexity games on several known induced path convexities: the monophonic $\mathrm{m}$-convexity and the $\ell_k$-convexities, based on induced paths and on induced paths of size at most $k$. We prove that the hull games $\mathrm{CHG}_{\mathrm{m}}$, $\mathrm{CHG}_{\ell_k}$ and their misère variants are PSPACE-complete for every $k\ge2$ even in graphs with diameter at most 3. We also use the Sprague-Grundy Theory to obtain a polynomial time algorithm to decide the winner of the games $\mathrm{CHG}_{\mathrm{m}}$ and $\mathrm{CHG}_{\ell_k}$ for any $k\ge2$ in disjoint unions of paths and cycles. For $k\ge3$ odd, we prove that Alice (1st player) wins $\mathrm{CHG}_{\ell_k}$ in the path $P_n$ if and only if $n$ is odd and she wins in the cycle $C_n$ if and only if $n=3$ or $n=\alpha\cdot (k+1)-1$ with $\alpha\ge2$. For $k\ge2$ even, the only periodic nimber sequences of $\mathrm{CHG}_{\ell_k}$ obtained through extensive computational testing occurred for $k=2^h-4$ with $h\ge3$, e.g, $k\in\{4,12,28,60,\ldots\}$. In this case ($k=2^h-4$ with $h\ge3$), we prove that the nimber sequences of $\mathrm{CHG}_{\ell_k}$ in $P_n$ and in $C_n$ are periodic and Alice loses (resp. wins) in $P_n$ (resp. $C_n$) with $n>k$ only when $n=3k+4$ (resp. $n\in\{2k+1,5k+3\}$). Finally, we show that, for $k=2$, the game $\mathrm{CHG}_{\ell_2}$ in paths $P_n$ is closely related to the classical game \emph{Couples-are-Forever} of J. H. Conway: it is still an open problem if the nimber sequence is periodic or not and Alice loses only for 12 values of $n$ up to $200$ million.
\end{abstract}

\section{Introduction}
 
In 1981, Frank Harary~\cite{harary81} published the first paper on graph convexities and, in 1984, he proposed the first graph convexity games~\cite{harary84,harary84b}.
The \emph{hull game} ($\CHG_{\gc}$) and the \emph{interval game} ($\CIG_{\gc}$) are among the first investigated games~\cite{buckley85,haynes-2003,nec-1988}, where $\gc$ indicates the geodesic convexity.
Despite that, only recently the hull game $\CHG_{\gc}$ was proved PSPACE-complete~\cite{araujo24}, while the PSPACE-hardness of the interval game $\CIG_{\gc}$ is still an open problem.
In this paper, we start the research of these games in the monophonic $\mc$-convexity and the $\ell_k$-convexities.
 
Every convexity $\mathcal{C}$ on a finite graph $G$ can be defined through an interval function $\I_\mathcal{C}(\cdot)$ in such a way that $\I_\mathcal{C}(\emptyset)=\emptyset$, $S\subseteq\I_\mathcal{C}(S)$ and $S\subseteq S'$ implies $\I_\mathcal{C}(S)\subseteq\I_\mathcal{C}(S')$ for every $S\subseteq S'\subseteq V(G)$. We say that $S$ is $\mathcal{C}$-convex if and only if $\I_\mathcal{C}(S)=S$.
The most investigated graph convexities are the path convexities, which are based on a given family $\mathcal{P}$ of paths of the graph such that $\I_\mathcal{C}(S)$ consists of $S$ and every vertex in a path of $\mathcal{P}$ between two vertices of $S$. As an example, in the geodesic convexity~\cite{harary81}, the monophonic convexity~\cite{jamison82}, the $m_k$ convexities~\cite{dragan99} and the $\ell_k$ convexities~\cite{gutierrez-protti-tondato2023}, the paths considered are the shortest paths, the induced paths, the induced paths of at least $k$ edges and the induced paths of at most $k$ edges, respectively.
The \emph{$\mathcal{C}$-convex hull} $\hull_\mathcal{C}(S)$ of $S\subseteq V(G)$ is the minimum $\mathcal{C}$-convex set that contains $S$. It is known that $\hull_\mathcal{C}(S)$ can be obtained by successive applications of the interval function $I_\mathcal{C}(\cdot)$ until obtaining a $\mathcal{C}$-convex set.
We say that $S\subseteq V(G)$ is an \emph{interval set} (resp. \emph{hull set}) in the convexity $\mathcal{C}$ if $\I_\mathcal{C}(S)=V(G)$ (resp. $\hull_\mathcal{C}(S)=V(G)$).
For more information, we refer the reader to the recent book~\cite{araujo-livro25}.
 
In the games $\CHG_\mathcal{C}$ (\emph{hull game}) and $\CIG_\mathcal{C}$ (\emph{interval game}) in a graph $G$, two players, Alice and Bob, starting with Alice, alternate turns selecting a \emph{playable} vertex of $G$, where the notion of a playable vertex depends on the game.
In $\CHG_\mathcal{C}$ (resp. $\CIG_\mathcal{C}$), a vertex is \emph{playable} if it does not belong to $\hull_\mathcal{C}(S)$ (resp. $\I_\mathcal{C}(S)$), where $S$ is the set of vertices selected before the current turn.
Therefore, the game $\CHG_\mathcal{C}$ (resp. $\CIG_\mathcal{C}$) ends when the set $S$ of selected vertices is a hull set (resp. an interval set), that is, when $\hull_\mathcal{C}(S) = V(G)$ (resp. $\I_\mathcal{C}(S) = V(G)$). Here we consider the normal variant: the player unable to move loses, that is, except in the empty game, the last to play wins.
 
Note that a game in normal variant has no draw and then, from the classic Zermelo's Theorem~\cite{zermelo13}, either Alice or Bob has a winning strategy.
So, the decision problem of these games is whether Alice wins (has a winning strategy).
 
In 1985, Buckley and Harary introduced the (geodesic) games $\CHG_{\gc}$ and $\CIG_{\gc}$, solving them for simple graph classes, such as paths, cycles and wheels. Only in 2024, Araujo et al. \cite{araujo24} solved them for trees, using the Sprague-Grundy Theory, and proved the PSPACE-hardness of $\CHG_{\gc}$. In 2026, Dailly et al. \cite{dailly26} extended the result of trees to cacti and other graph classes. The PSPACE-hardness of $\CIG_{\gc}$ is still an open problem. Curiously, Araujo et al. \cite{benevides26} recently proved the PSPACE-hardness of the optimization variant of $\CIG_{\gc}$, while the PSPACE-hardness of the optimization variant of $\CHG_{\gc}$ is still open.
 
In this paper, we investigate these games in several convexities of induced paths. We prove that the hull games $\CHG_{\mc}$, $\CHG_{\ell_k}$ and their misère variants are PSPACE-complete for every $k\ge2$ even in graphs with diameter at most three. 
We show that the PSPACE-hardness result obtained is also valid for the geodesic hull game $\CHG_{\gc}$ and is much simpler than the one in~\cite[Thm. 3.2]{araujo24}.
 
We also use the Sprague-Grundy Theory to obtain a polynomial time algorithm to decide the winner of the games $\CHG_{\mc}$ and $\CHG_{\ell_k}$ for any $k\ge2$ in disjoint unions of paths and cycles.
In such graphs, we prove that Alice wins $\CHG_{\mc}$ if and only if the number of odd paths (in the number of vertices) and triangles (cycles with 3 vertices) is odd.
 
For $k\ge3$ odd, we prove that Alice wins $\CHG_{\ell_k}$ in the path $P_n$ if and only if $n$ is odd and she wins in the cycle $C_n$ if and only if $n=3$ or $n=\alpha\cdot (k+1)-1$ with $\alpha\ge2$.
For $k\ge2$ even, the only periodic nimber sequences of $\CHG_{\ell_k}$ obtained through extensive computational testing occurred for $k=2^h-4$ with $h\ge3$, e.g, $k\in\{4,12,28,60,\ldots\}$. In this case ($k=2^h-4$ with $h\ge3$), we prove that the nimber sequences of $\CHG_{\ell_k}$ in $P_n$ and in $C_n$ are periodic and Alice loses (resp. wins) in $P_n$ (resp. $C_n$) with $n>k$ only when $n=3k+4$ (resp. $n\in\{2k+1,5k+3\}$).
The periodic sequences have long transients and long periodic tails; therefore, the proofs of periodicity are lengthy and highly technical.
 
Finally, we show that, for $k=2$, the game $\CHG_{\ell_2}$ in paths $P_n$ is closely related to the classical game \emph{Couples-are-Forever} of J. H. Conway.
In this game, there is one heap with coins and a move consists in selecting a heap with more than two coins and breaking it into two smaller heaps. At the end, all heaps have size 1 or 2 and the player unable to move loses.
It is still an open problem if the nimber sequence of \emph{Couples-are-Forever} is periodic or not.
In 1999, Caines et al.~\cite{caines99} obtained the nimber sequence of \emph{Couples-are-Forever} up to $50$ million, showing that Alice loses for only 14 values of $n$ up to $50$M:
\[
1,\ 2,\ 5,\ 13,\ 21,\ 31,\ 47,\ 73,\ 99,\ 125,\ 151,\ 177,\ 315,\ 409.
\]
We prove that, up to a shift of one position, the nimber sequence of \emph{Couples-are-Forever} and the nimber sequences of the game $\CHG_{\ell_2}$ on paths $P'_n$ and $P''_n$ (considering that one or two endpoints have already been selected at the start of the game, respectively) are identical.
Moreover, we prove that the nimber of $\CHG_{\ell_2}$ on the path $P_n$ is always equal to $CF_{n+3}$ or to its predecessor $CF_{n+2}$, where $CF$ denotes the nimber sequence of \emph{Couples-are-Forever}. We also extended the computations of~\cite{caines99} up to $n = 200$ million \cite{rudini26cf}. Once the three-position shift is taken into account, the two sequences do not coincide, but they differ in only $54{,}412$ of the first $200$ million positions, an agreement rate of $99.9728\%$. Moreover, with $n$ up to $200$ million, Alice loses $\CHG_{\ell_2}$ on the path $P_n$ for only 12 values of $n$:
\[
2,\ 10,\ 18,\ 28,\ 44,\ 70,\ 96,\ 122,\ 148,\ 174,\ 312,\ 406,
\]
and she wins in the cycle $C_n$ for only 13 values of $n$:
\[
3,\ 5,\ 13,\ 21,\ 31,\ 47,\ 73,\ 99,\ 125,\ 151,\ 177,\ 315,\ 409.
\]

\section{Hull games are PSPACE-hard}
 
In this section, we prove that the normal variant of hull games $\CHG_{\mc}$ and $\CHG_{\ell_k}$ are PSPACE-complete for every $k\ge2$. Here we consider the games as decision problems: given a graph, does Alice have a winning strategy?
Since the number of turns is at most $n$ and, in each turn, the number of possible vertices to select is at most $n$, all these games are polynomially bounded two player games, which implies that they are in PSPACE~\cite{demaine-2009}.
 
We obtain a reduction from the \textsc{Clique Forming} game, which is PSPACE-complete~\cite{schaefer-1978} and was used recently in reductions of other graph convexity games~\cite{araujo24, araujo-cocoon25, chandran24}.
In this game on a graph $G$, Alice and Bob alternately select vertices and the subset of the chosen vertices must induce a clique.
 
Although the geodesic hull game $\CHG_{\gc}$ was proven to be PSPACE-complete in~\cite{araujo24}, we also include it in the theorem below, given that the following proof is much simpler.
Recall that the games below are in normal variant (the last to play wins).
 
\begin{theorem}\label{teo-pspace1}
The games $\CHG_{\ell_k}$, $\CHG_{\mc}$ and $\CHG_{\gc}$ are PSPACE-complete even in graphs with diameter at most three for every $k\ge2$.
\end{theorem}
 
\begin{proof}
First consider the monophonic $\mc$-convexity or some $\ell_k$-convexity with $k\ge3$. The result of the $\ell_2$-convexity will be shown later.
Let $H$ be an instance of the \textsc{Clique Forming} game. We may assume that $H$ is not complete.
Let $G$ be the graph obtained from $H$ by adding five new vertices $u_1$, $u_1'$, $u_2$, $u_2'$ and $w$ and including every edge from $H$ to $\{u_1,u_1',u_2,u_2'\}$ and the edges $u_1u_1'$, $u_2u_2'$, $u_1w$ and $u_1'w$. See Figure~\ref{fig-pspace}.
 
\begin{figure}[!bt]
\centering\scalebox{1.0}{
\begin{tikzpicture}[scale=1]
\tikzstyle{vertex}=[draw,circle,fill=black!10,minimum size=15pt,inner sep=1pt]
 
\node[vertex] (w) at (-0.5,1.0) {$w$};
\node[vertex] (u1) at (0.5,1.7) {$u_1$};
\node[vertex] (u1b)at (0.5,0.3) {$u_1'$};
\node[vertex] (u2) at (5.5,1.7) {$u_2$};
\node[vertex] (u2b)at (5.5,0.3) {$u_2'$};
 
\draw (3,1) ellipse (0.5cm and 1cm);
\node at (3,1) {\textcolor{red}{$H$}};
\node at (1.7,2.3) {\textcolor{blue}{$G$}};
 
\draw (u1)--(u1b)--(w)--(u1); \draw (u2)--(u2b);
\path[-]
(u1)edge(2.4,1)edge(2.5,0.2)edge(2.5,1.8)edge(2.4,1.4)edge(2.4,0.6);
\path[-]
(u2)edge(3.6,1)edge(3.5,0.2)edge(3.5,1.8)edge(3.6,1.4)edge(3.6,0.6);
\path[-]
(u1b)edge(2.4,0.9)edge(2.5,0.1)edge(2.5,1.7)edge(2.4,1.3)edge(2.4,0.5);
\path[-]
(u2b)edge(3.6,0.9)edge(3.5,0.1)edge(3.5,1.7)edge(3.6,1.3)edge(3.6,0.5);
\end{tikzpicture}}
\caption{Reduction of Theorem~\ref{teo-pspace1} from the \textsc{Clique Forming} game to $\CHG_{\mc}$ and $\CHG_{\ell_k}$ with $k\ge3$, where $H$ is not complete.}
\label{fig-pspace}
\end{figure}

Notice that the only induced paths of length greater than two are inside $H$ or are of type $wu_1''vu_2''$, where $v\in V(H)$, $u_1''\in\{u_1,u_1'\}$ and $u_2''\in\{u_2,u_2'\}$.
 
Let $X\in\{Alice, Bob\}$ be a player with winning strategy in the clique-forming game on $H$, and let $Y$ be its opponent.
Then $X$ always follows the winning strategy in $H$, playing in the same vertices, unless $Y$ plays outside $H$.
If $Y$ plays in $w$, then $X$ plays in $u_2$ and wins immediately.
If $Y$ plays in $u_2$ or $u_2'$, then $X$ plays in $w$ and wins immediately.
If $Y$ plays in $u_1$ (resp. $u_1'$), then $X$ plays in $u_1'$ (resp. $u_1$).
If $Y$ plays in a vertex of $H$ not adjacent to the other selected vertices, then $X$ plays in $w$ and wins immediately, since $u_1$ and $u_2$ are taken after $Y$'s move and they take all vertices of $G-w$.
 
Since $X$ wins the clique forming game in $H$, it follows that $Y$ is the first to play in a vertex of $\{w,u_2,u_2'\}$ or in a vertex of $H$ not adjacent to the other selected vertices. This implies that the game ends after the next move of $X$, which is the last to play and wins the hull game in the normal variant of the considered convexity.
 
Curiously, although this proof does not work to the $\ell_2$-convexity, it is also valid for the geodesic hull game $\CHG_{\gc}$ and is much simpler than the one of~\cite[Thm. 3.2]{araujo24}, shown in Figure~\ref{fig-pspace2}.
However, the PSPACE-hardness proof of~\cite[Thm. 3.2]{araujo24} is very useful and is valid to the game $\CHG_{\ell_2}$ of the $\ell_2$-convexity, since the graph constructed there (see Figure~\ref{fig-pspace2}) has diameter two and the relevant induced paths, which are the geodesics, have length two.
 
\begin{figure}[!bt]
\centering\scalebox{1.0}{
\begin{tikzpicture}[scale=1]
\tikzstyle{vertex}=[draw,circle,fill=black!10,minimum size=15pt,inner sep=1pt]
 
\node[vertex] (u1) at (0.5,1.7){$u_1$};
\node[vertex] (u1b)at (0.5,0.3){$u_1'$};
\node[vertex] (u2) at (5.5,1.7){$u_2$};
\node[vertex] (u2b)at (5.5,0.3){$u_2'$};
\node[vertex] (w1) at (2,3)    {$w_1$};
\node[vertex] (w1b)at (2,4.2)  {$w_1'$};
\node[vertex] (w2) at (4,3)    {$w_2$};
\node[vertex] (w2b)at (4,4.2)  {$w_2'$};
\draw (u1b)--(u1)--(w1)--(w2)--(u2)--(u2b);
\draw (u1)--(w1b)--(w2);
\draw (u2)--(w2b)--(w1);
\draw (w1)--(w1b)--(w2b)--(w2);
\draw (w1) to[out=150,in=150] (u1b);
\draw (w1b) to[out=180,in=180] (u1b);
\draw (w2) to[out=30,in=30] (u2b);
\draw (w2b) to[out=0,in=0] (u2b);
\draw (3,1) ellipse (0.5cm and 1cm);
\node at (3,1) {\textcolor{red}{$H$}};
\node at (3,2.5) {\textcolor{blue}{$G$}};
\path[-]
(u1)edge(2.4,1)edge(2.5,0.2)edge(2.5,1.8)edge(2.4,1.4)edge(2.4,0.6);
\path[-]
(u2)edge(3.6,1)edge(3.5,0.2)edge(3.5,1.8)edge(3.6,1.4)edge(3.6,0.6);
\path[-]
(u1b)edge(2.4,0.9)edge(2.5,0.1)edge(2.5,1.7)edge(2.4,1.3)edge(2.4,0.5);
\path[-]
(u2b)edge(3.6,0.9)edge(3.5,0.1)edge(3.5,1.7)edge(3.6,1.3)edge(3.6,0.5);
\end{tikzpicture}}
\caption{Reduction of~\cite[Thm. 3.2]{araujo24} used in Theorem~\ref{teo-pspace1} from the \textsc{Clique Forming} game to $\CHG_{\ell_2}$, where $H$ is not complete.}
\label{fig-pspace2}
\end{figure} 
\end{proof}
 
The following theorem is a direct consequence of the PSPACE-hardness of misère $\CHG_{\gc}$ in~\cite[Thm. 3.1]{araujo24}. In the misère variant of a game, the last to play loses.
 
\begin{theorem}\label{teo-pspace2}
The games misère $\CHG_{\ell_k}$ and misère $\CHG_{\mc}$ are PSPACE-complete even in graphs with diameter at most two for every $k\ge2$.
\end{theorem}
 
\begin{proof}
It follows directly from~\cite[Thm. 3.1]{araujo24}, where the reduction from the \textsc{Clique Forming} game is shown in Figure~\ref{fig-pspace3}. Although the graph $H$ (which is not complete) may contain long induced paths (which are not relevant to the proof of~\cite[Thm. 3.1]{araujo24}), such paths imply that their endpoints are non-adjacent; consequently, every vertex of $G$ can be taken via induced paths of length two involving the universal vertices.
 
\begin{figure}[!bt]
\centering\scalebox{1.0}{
\begin{tikzpicture}[scale=1]
\tikzstyle{vertex}=[draw,circle,fill=black!10,minimum size=15pt,inner sep=1pt]
 
\node[vertex] (u1) at (0.5,1.7) {$u_1$};
\node[vertex] (u1b)at (0.5,0.3) {$u_1'$};
\node[vertex] (u2) at (5.5,1.7) {$u_2$};
\node[vertex] (u2b)at (5.5,0.3) {$u_2'$};
 
\draw (3,1) ellipse (0.5cm and 1cm);
\node at (3,1) {\textcolor{red}{$H$}};
\node at (1.7,2.3) {\textcolor{blue}{$G$}};
 
\draw (u1)--(u1b); \draw (u2)--(u2b);
\path[-]
(u1)edge(2.4,1)edge(2.5,0.2)edge(2.5,1.8)edge(2.4,1.4)edge(2.4,0.6);
\path[-]
(u2)edge(3.6,1)edge(3.5,0.2)edge(3.5,1.8)edge(3.6,1.4)edge(3.6,0.6);
\path[-]
(u1b)edge(2.4,0.9)edge(2.5,0.1)edge(2.5,1.7)edge(2.4,1.3)edge(2.4,0.5);
\path[-]
(u2b)edge(3.6,0.9)edge(3.5,0.1)edge(3.5,1.7)edge(3.6,1.3)edge(3.6,0.5);
\end{tikzpicture}}
\caption{Reduction of~\cite[Thm. 3.1]{araujo24} used in Theorem~\ref{teo-pspace2} from the \textsc{Clique Forming} game to misère $\CHG_{\mc}$ and misère $\CHG_{\ell_k}$ with $k\ge2$, where $H$ is not complete.}
\label{fig-pspace3}
\end{figure} 
 
\end{proof}
 
\section{Hull games are polynomial in paths and cycles}\label{sec:poly}
 
We first summarize the Sprague-Grundy Theory regarding finite impartial games in the normal variant.
We say that a game is \emph{finite} if it always ends after a finite number of moves.
The \emph{length} of a position of a finite combinatorial game is the maximum possible number of moves until the end of the game.
We say that a game is \emph{impartial} if the possible moves in any position of the game do not depend on the current player.
The \emph{nimber} of a position of a finite impartial game is recursively defined (on the length of the position) as a natural number associated in the following way.
If the position has length $0$ (and then there are no possible moves and the 1st player loses), then its nimber is $0$.
Otherwise, its nimber is the \emph{minimum excluded value} $\mex\{N_1,\ldots,N_\ell\}$, where $\{N_1,\ldots,N_\ell\}$ is the set of all nimbers of the positions that can be achieved after one move and the \emph{minimum excluded value} ($\mex$) is the smallest natural number not belonging to the set.
 
To state the next theorem, we also need the definition of sum of games.
Given disjoint positions $J_1,\ldots,J_k$ of finite impartial games, the \emph{sum} $J_1+\ldots+J_k$ is the game position in which each player, when making their move, chooses one of the $k$ positions (say $J_i$) and makes a move in $J_i$.
It is easy to see that the sum of games is commutative and associative.
 
\begin{theorem}[Sprague-Grundy~\cite{grundy39,sprague36}]\label{thm-sprague}
The 1st player wins an impartial finite game in the normal variant if and only if the nimber is positive.
Moreover, the nimber of the sum of $k$ disjoint positions with nimbers $n_1,\ldots,n_k$ is equal to $n_1\xor\ldots\xor n_k$, where $\xor$ is the bitwise xor (exclusive-or) operation.
\end{theorem}

With this in hand, we can determine the winner of the hull games $\CHG_{\mc}$ and $\CHG_{\ell_k}$ (which are in normal play) for paths and cycles by calculating its nimbers.
We write $\nim_{\mc}(G)$ and $\nim_{\ell_k}(G)$ for the nimbers of a graph $G$ in the games $\CHG_{\mc}$ and $\CHG_{\ell_k}$, respectively.

We remark that, in paths and cycles, the hull game and the interval game coincide for the monophonic and the $\ell_k$-convexities, since in these graphs the interval of any set of vertices is already convex. Thus all results of this paper on paths and cycles also hold for the interval games $\CIG_{\mc}$ and $\CIG_{\ell_k}$.
  
We start with the (monophonic) $\mc$-convexity.
Let $P'_n$ be the game considering that exactly one endpoint is already selected in the start of the game.
 
\begin{theorem}\label{thm-nimber-monof}
Let $n\ge1$.
Then the nimber $\nim_{\mc}(P'_n)=n-1$.
Moreover, the nimber $\nim_{\mc}(P_n)=1$, if $n$ is odd, and $\nim_{\mc}(P_n)=0$, otherwise. Finally, $\nim_{\mc}(C_3)=1$ and $\nim_{\mc}(C_n)=0$ if $n\ge4$.
\end{theorem}
 
\begin{proof}
Notice that $\nim_{\mc}(P'_1)=0$, since there is no move. So, let $v_1,\ldots v_n$ be the vertices of the path $P'_n$, assuming that $v_n$ is already selected in the start of the game. Considering all possible moves for the 1st player, in the vertices $v_1,\ldots,v_{n-1}$, we have
\[
\nim_{\mc}(P'_n)\ =\ \mex\{\nim_{\mc}(P'_1),\ldots,\nim_{\mc}(P'_{n-1})\}.
\]
Then, by induction, $\nim_{\mc}(P'_n)=\mex\{0,\ldots,n-2\}=n-1$.
Also note that, by playing in the vertex $v_{i+1}$ with $i=0,\ldots,n-1$,
\[
\nim_{\mc}(P_n)\ =\ \mex\{\nim_{\mc}(P'_{i+1})\xor\nim_{\mc}(P'_{j+1}):\ i+j=n-1,\ \ i,j\ge0\}.
\]
Then
\[
\nim_{\mc}(P_n)\ =\ \mex\{i\xor j:\ i+j=n-1,\ \ i,j\ge0\}.
\]
If $n$ is even, then $i+j$ is odd and, consequently, $i\ne j$ and $i\xor j\ne 0$, implying that $\nim_{\mc}(P_n)=0$.
If $n$ is odd, then $i+j$ is even (have the same parity) and consequently $i\xor j=0$ when $i=j=(n+1)/2$ and $i\xor j$ is never equal to $1$, implying that $\nim_{\mc}(P_n)=1$.
 
Finally, it is easy to check that $\nim_{\mc}(C_3)=1$ and that $\nim_{\mc}(C_n)=0$ with $n\ge4$, since the first player always loses (the opponent just select a non-adjacent vertex of the 1st selected vertex).
\end{proof}

Now we consider the $\ell_k$-convexities. Let $P''_n$ be the game considering that both endpoints are already selected in the start of the game.
 
\begin{theorem}\label{thm-nimber-Pn}
Let $k\ge2$.
The nimbers $\nimp{n}$, $\nimpb{n}$ and $\nimpc{n}$ in the game $\CHG_{\ell_k}$ can be determined in quadratic time by the following recurrences: $\nimp{1}=1$, $\nimp{2}=0$, $\nimpb{1}=0$, $\nimpb{2}=1$ and $\nimpc{n}=0$ for $n=1,\ldots,k+1$. Moreover,
\begin{align*}
\nimp{n}\ =\ &\mex\Big\{\nimpb{i} \xor\nimpb{n-i+1}\ :\ i=1\ldots\lceil n/2\rceil\Big\},\ \mbox{for $n\ge3$};\\
\nimpb{n}\ =\ &\mex\Big\{\nimpc{i} \xor\nimpb{n-i+1}\ :\ i=2\ldots n\Big\},\ \ \ \ \ \ \mbox{for $n\ge3$};\\
\nimpc{n}\ =\ &\mex\Big\{\nimpc{i} \xor\nimpc{n-i+1}\ :\ i=2\ldots\lceil n/2\rceil\Big\},\ \mbox{for $n\ge k+2$}.
\end{align*}
\end{theorem}
 
\begin{proof}
Clearly, $\nimpb{1}=0$ and $\nimpc{n}=0$ if $n\in\{2,\ldots,k+1\}$, since there is no move. Moreover, $\nimp{1}=\mex\{0\}=1$, $\nimp{2}=\mex\{1\}=0$ and $\nimpb{2}=\mex\{0\}=1$.
So, let $v_1,\ldots v_n$ be the vertices of the path.
 
First consider the game in $P_n$ with $n\ge3$.
By symmetry, we may consider that the first move is in a vertex $v_i$ for $i\in\{1,\ldots,\lceil n/2\rceil\}$.
This move produces two independent games, $P'_{i}$ and $P'_{n-i+1}$, and we are done by the Sprague-Grundy Theorem~\ref{thm-sprague}.
 
Now consider the game in $P'_n$ with $n\ge3$ and $v_1$ already selected. Then the possible moves are in a vertex $v_i$ among $v_2,\ldots,v_n$. This move produces two independent games, $P''_i$ and $P'_{n-i+1}$, and we are done from Sprague-Grundy Theorem~\ref{thm-sprague}.
 
Finally consider the game in $P''_n$ with $n\ge k+2$ and $v_1$ and $v_n$ already selected. Then the possible moves are in the vertices $v_2,\ldots,v_{n-1}$. By symmetry, we may assume that the move is in a vertex $v_i$ with $i\in\{2,\ldots,\lceil n/2\rceil\}$. This move produces two independent games, $P''_i$ and $P''_{n-i+1}$, and we are done from Sprague-Grundy Theorem~\ref{thm-sprague}.
 
With a dynamic programming algorithm, it is possible to determine $\nimp{n}$, $\nimpb{n}$ and $\nimpc{n}$ in time $O(n^2)$.
\end{proof}
 
We can also determine the winner for any cycle $C_n$.
 
\begin{corollary}\label{cor-nimber-Cn}
Let $k\ge2$.
The nimber $\nim(C_n)$ of the cycle $C_n$ with $n\ge3$ vertices in the normal variant of the game $\CHG_{\ell_k}$ is $1$ for $n=3$, is $0$ for $4\le n\le2k$, and is $\mex\{\nimpc{n+1}\}$, otherwise. Therefore, its value is either $0$ or $1$, and can be determined in quadratic time.
\end{corollary}
 
\begin{proof}
Clearly $\nim_{\ell_k}(C_3)=1$ for any $k\ge2$, since all the three vertices will be selected during the game.
Moreover, if $4\le n\le2k$, then $\nim_{\ell_k}(C_n)=0$, since, after Alice's first move, Bob can select the farthest vertex, winning the game immediately.
So let $n>2k$. Notice that the first move of Alice generates the game in $P''_{n+1}$, by duplicating the chosen vertex and considering them as the endpoints of $P''_{n+1}$ (recall that $n>2k$ and then all internal vertices of $P''_{n+1}$ are playable, just as in $C_n$ after Alice's first move).
From the Sprague--Grundy Theorem~\ref{thm-sprague}, $\nim(C_n)=\mex\{\nim(P''_{n+1})\}$.
From Theorem~\ref{thm-nimber-Pn}, it is possible to determine $\nim_{\ell_k}(C_n)$ in time $O(n^2)$.
\end{proof}

\section{Nimbers of the hull game in convexity $\ell_k$}
 
The previous section showed how to compute the nimbers $\nimp{n}$ and $\nim_{\ell_k}(C_n)$ for every $k,n\ge2$. In this section, we show some basic facts regarding $\nimp{n}$, $\nimpb{n}$ and $\nimpc{n}$.
 
Theorem~\ref{thm-nimber-Pn} shows the values of $\nimpc{n}=0$ for $n=1,\ldots,k+1$.
The next lemma shows the values for $n=k+2,\ldots,3k+2$.
 
\begin{lemma}\label{lem:phaseA}
Let $k\ge2$ and $1\le n\le 3k+2$. Then
\[
\nimpc{n}=
\begin{cases}
0, &\mbox{if $n<k+1$},\\
n\%(k+1), &\mbox{if $n\ge k+1$},
\end{cases}
\]
where $\%$ is the remainder operator.
\end{lemma}
 
\begin{proof}
From Theorem~\ref{thm-nimber-Pn}, $\nimpc{n}=0$ for $n\leq k+1$. So, let $n>k+1$ and let $r=n\%(k+1)$.
The $\mex$ function of the recurrence of $\nimpc{n}$ in Theorem~\ref{thm-nimber-Pn} includes all terms
$\nimpc{i}\xor\nimpc{j}$ such that $i+j=n+1$ and $i\ge j\ge2$.
 
First assume that $n\in\{k+2,\ldots,2k+1\}$.
Since $i+j=n+1\leq 2k+2$ and $i\ge j$, it follows that $2\le j\le k+1$ and $\nimpc{j}=0$. Then the $\mex$ function exactly contains the $k$ values $\nimpc{n-1},\ldots,\nimpc{n-k}$ and, by induction, $\nimpc{n}=\mex(\{0,\ldots,r-1\})=r$.
 
Now assume that $n=2k+2$. Then,
\[
\nimpc{n}\ =\ \nimpc{2k+2}\ =\ \mex\{\nimpc{k+2},\ldots,\nimpc{2k+1}\}\ =\ \mex\{1,\ldots,k\}=0.
\]
 
Finally, assume that $n\in\{2k+3,\ldots,3k+2\}$.
As before, the $\mex$ function for $\nimpc{n}$ contains the $k$ values $\nimpc{n-1},\ldots,\nimpc{n-k}$, which are by induction $\{0,\ldots,k\}\setminus\{r\}$.
We have to show that $r$ does not appear in the $\mex$ function. It also contains the values $\nimpc{i}\xor\nimpc{j}$ with $i+j=n+1$ and $i\ge j\ge k+2$, which are by induction $r_i\xor r_j$, where $r_i=i\%(k+1)$ and $r_j=j\%(k+1)$.
Since $j\ge k+2$: $i\le 2k+1$.
Therefore $r_i\xor r_j$ has the same parity of $i\xor j$, which is the same parity of $n+1$ and also of $(n+1)\%(k+1)$, because $n\in\{2k+3,\ldots,3k+2\}$. Consequently $r_i\xor r_j$ cannot be $r=n\%(k+1)$.
\end{proof}
 
\begin{theorem}\label{thm:shift}
For every $k\ge2$ and every $n\ge1$,
\[
\nimpb{n} = \nimpc{n+k}.
\]
\end{theorem}
 
\begin{proof}
We argue by induction on $n$. First notice that $\nimpb{1}=\nimpc{k+1}=0$ and $\nimpb{2}=\nimpc{k+2}=1$ from Theorem~\ref{thm-nimber-Pn} and Lemma~\ref{lem:phaseA}. First assume that $3\le n\le k+1$.
Since $\nimpc{i}=0$ for every $i=2,\ldots,k+1$, we have by induction and Theorem~\ref{thm-nimber-Pn},
\[
\nimpb{n}\ =\ \mex\{\nimpb{1},\ldots,\nimpb{n-1}\}\ =\ \mex\{0,\ldots,n-2\}\ =\ n-1,
\]
and consequently $\nimpb{n}=(n-1)\%(k+1)=(n+k)\%(k+1)=\nimpc{n+k}$, since $n\leq k+1$.
 
Finally assume that $n\ge k+2$.
From Theorem~\ref{thm-nimber-Pn}, $\nimpb{n}=\mex_{i=2}^{n}\{\nimpc{i}\ \xor\ \nimpb{n-i+1}\}$. Since $\nimpb{n-i+1}=\nimpc{n-i+k+1}$ for $i=2,\ldots,n$ by induction:
\[
\nimpb{n}\ =\ \mex_{i=2}^{n}\{\nimpc{i}\xor\nimpc{n-i+k+1}\}.
\]
Since $n>\lceil(n+k)/2\rceil$, we may write the recurrence of Theorem~\ref{thm-nimber-Pn} for $\nimpc{n+k}$ as
\[
\nimpc{n+k}\ =\ \mex_{i=2}^{n}\{\nimpc{i}\xor\nimpc{n-i+k+1}\},
\]
and we are done.
\end{proof}
 
\section{Hull game in convexity $\ell_2$ and the game Couples-are-Forever}\label{sec:forever}
 
In this section, we show that, for $k=2$, the game $\CHG_{\ell_2}$ in paths $P_n$ is closely related to the classical game \emph{Couples-are-Forever} of J. H. Conway~\cite{caines99}.
In this game, there is one heap with coins and a move consists in selecting a heap with more than two coins and breaking it into two smaller heaps.
At the end, all heaps have size 1 or 2.
Let $CF_n$ be the nimber of the game Couples-are-Forever starting with one heap with $n$ coins.
Then $CF_1=CF_2=0$ and, by the Sprague-Grundy Theorem~\ref{thm-sprague},
\[
CF_n=\mex\{CF_i\xor CF_{n-i}\ :\ i=1\ldots\lfloor n/2\rfloor\}.
\]
It is still an open problem if the nimber sequence of \emph{Couples-are-Forever} is periodic or not~\cite{caines99} and Alice only loses to 14 values of $n$ up to $200$ million:
\[
1,\ 2,\ 5,\ 13,\ 21,\ 31,\ 47,\ 73,\ 99,\ 125,\ 151,\ 177,\ 315,\ 409.
\]
 
We first prove that, up to a shift of one position, the nimber sequence of \emph{Couples-are-Forever} is identical to the nimber sequence of the game $\CHG_{\ell_2}$ on paths $P'_n$ and $P''_n$ (considering that one or two endpoints have already been selected at the start of the game, respectively).
 
\begin{theorem}\label{thm.forever}
Let $n\ge1$. Then $\nimptb{n}=CF_{n+1}$ and $\nimptc{n+1}=CF_{n}$. Moreover, $\nimpt{n}\le CF_{n+3}$ and $\nimpt{n}\in\{CF_{n+2},CF_{n+3}\}$.
\end{theorem}
 
\begin{proof}
From Theorem~\ref{thm-nimber-Pn}, $\nim_{\ell_2}(P''_2)=CF_1=0$, $\nim_{\ell_2}(P''_3)=CF_2=0$ and, for $n\ge 3$,
\[
\nimptc{n+1}\ =\ \mex\Big\{\nimptc{i+1} \xor\nimptc{n-i+1}\ :\ i=1\ldots\lfloor n/2\rfloor\Big\}.
\]
Then, by induction,
\[
\nimptc{n+1}\ =\ \mex\Big\{CF_i \xor CF_{n-i}\ :\ i=1\ldots\lfloor n/2\rfloor\Big\}\ =\ CF_{n}.
\]
Also, from Theorem~\ref{thm:shift}, $\nimptb{n}=\nimptc{n+2}=CF_{n+1}$.
Moreover, by Theorem~\ref{thm-nimber-Pn}, for $n\ge3$,
\[
\nimpt{n}\ =\ \mex\Big\{\nimptb{i-1}\xor\nimptb{n-i+2}\ :\ i=2,\ldots,\lceil n/2\rceil+1\Big\},
\]
and therefore, since $\lfloor (n+3)/2\rfloor=\lceil n/2\rceil+1$,
\[
\nimpt{n}\ =\ \mex\Big\{CF_{i}\xor CF_{n-i+3}\ :\ i=2,\ldots,\lfloor (n+3)/2\rfloor\Big\}.
\]
This formula also holds for $n=1,2$, since its right-hand side gives $\mex\{CF_2\xor CF_2\}=\mex\{0\}=1=\nimpt{1}$ and $\mex\{CF_2\xor CF_3\}=\mex\{1\}=0=\nimpt{2}$.
The set above is contained in the set $\{CF_i\xor CF_{n+3-i}:i=1,\ldots,\lfloor(n+3)/2\rfloor\}$ whose $\mex$ is $CF_{n+3}$, and the only term of the latter that may be missing from the former is the one with $i=1$, namely $CF_1\xor CF_{n+2}=CF_{n+2}$ (as $CF_1=0$). Since every value smaller than $\nimpt{n}$ belongs to the former set, it also belongs to the latter, so $\nimpt{n}\le CF_{n+3}$. Moreover, if $\nimpt{n}<CF_{n+3}$, then $\nimpt{n}$ belongs to the latter set but not to the former, so $\nimpt{n}=CF_{n+2}$.
\end{proof}
 
Notice in the proof above that $\nimpt{n}$ is very close to $CF_{n+3}$ (except for the value $i=1$ in the $\mex$ function).
After extensive computational experiments for $n$ up to $200$ million \cite{rudini26cf}, we observed that the nimber sequences of the original game $\CHG_{\ell_2}$ (on paths $P_n$) and \emph{Couples-are-Forever} are such that $\nimpt{n}=CF_{n+3}$ with probability $99.97\%$ (see Figure~\ref{fig.grafico}).
Moreover, this probability appears to stabilize as $n$ grows, as can be seen in Figure~\ref{fig.proporcao}.
The sequence $CF$ was computed up to $n=200$ million on the ParGO cluster of the Federal University of Ceará, on an instance with an AMD Ryzen 7 5700G processor (16 threads) and 55 GB of RAM, in about 15 days; the data and the programs are available at~\cite{rudini26cf}.

Summarizing, with $n\ge1$ up to $200$ million,
\begin{align*}
\nimptc{n}=0\ \iff\ &n\in\{1,2,3,6,14,22,32,48,74,100,126,152,178,316,410\}\\
\nim_{\ell_2}(C_n)>0\ \iff\ &n\in\{3,5,13,21,31,47,73,99,125,151,177,315,409\}\\
\nimptb{n}=0\ \iff\ &n\in\{1,4,12,20,30,46,72,98,124,150,176,314,408\}\\
\nimpt{n}=0\ \iff\ &n\in\{2,10,18,28,44,70,96,122,148,174,312,406\}
\end{align*}
 
In other words, for $n$ up to $200$ million, Alice loses $\CHG_{\ell_2}$ on the path $P_n$ for only 12 values of $n$ and she wins on the cycle $C_n$ for only 13 values of $n$.

The largest nimber found in these computations is $325$, attained by the sequence of \emph{Couples-are-Forever} at $CF_{19{,}739{,}544}=325$. Then no nimber of $\CHG_{\ell_2}$ on a path exceeds $325$ either. Moreover, $\nimpt{19{,}739{,}541}=325$.

\begin{figure}\centering
\begin{tikzpicture}
\begin{axis}[
    width=12cm, height=8cm,
    title={Cumulative divergences between $\nimpt{n}$ and $CF_{n+3}$},
    xlabel={$n$},
    ylabel={},
    xmin=0, xmax=204000000,
    ymin=0, ymax=56000,
    xtick={0,50000000,100000000,150000000,200000000},
    xticklabels={0,50M,100M,150M,200M},
    ytick={10000,20000,30000,40000,50000},
    yticklabels={10000,20000,30000,40000,50000},
    scaled x ticks=false,
    scaled y ticks=false,
    grid=both,
    grid style={gray!20},
    axis lines=left,
    axis line style={gray!60},
    tick label style={font=\small, gray!70!black},
    label style={font=\small},
    title style={font=\small\bfseries},
]
\addplot[
    color=gray!70!black,
    line width=1pt,
    mark=*,
    mark size=1.4pt,
    mark options={fill=gray!70!black, draw=none},
] coordinates {
    (0,    0)
    (5000000,    1395)
    (10000000,    2833)
    (15000000,    4256)
    (20000000,    5693)
    (25000000,    7038)
    (30000000,    8425)
    (35000000,    9824)
    (40000000,    11162)
    (45000000,    12513)
    (50000000,    13845)
    (55000000,    15131)
    (60000000,    16494)
    (65000000,    17831)
    (70000000,    19217)
    (75000000,    20569)
    (80000000,    21903)
    (85000000,    23284)
    (90000000,    24580)
    (95000000,    25921)
    (100000000,    27304)
    (105000000,    28701)
    (110000000,    30088)
    (115000000,    31404)
    (120000000,    32789)
    (125000000,    34170)
    (130000000,    35473)
    (135000000,    36742)
    (140000000,    38101)
    (145000000,    39482)
    (150000000,    40848)
    (155000000,    42210)
    (160000000,    43570)
    (165000000,    44929)
    (170000000,    46277)
    (175000000,    47595)
    (180000000,    48973)
    (185000000,    50290)
    (190000000,    51635)
    (195000000,    53000)
    (200000000,    54412)
};
\end{axis}

\end{tikzpicture}
\caption{\label{fig.grafico}Number of divergences between $\nimpt{n}$ and $CF_{n+3}$, $n$ up to $200$ million. Divergence rate $0.0272\%$ and agreement rate $99.97\%$.}
\end{figure}

\begin{figure}\centering
\begin{tikzpicture}
\begin{axis}[
    width=13cm, height=7.5cm,
    title={Proportion of divergences between $\nimpt{n}$ and $CF_{n+3}$},
    xlabel={$n$},
    ylabel={},
    xmin=0, xmax=204000000,
    ymin=0.0262, ymax=0.0322,
    xtick={0,50000000,100000000,150000000,200000000},
    xticklabels={0,50M,100M,150M,200M},
    ytick={0.0270,0.0280,0.0290,0.0300,0.0310,0.0320},
    yticklabel={\pgfmathprintnumber[fixed,fixed zerofill,precision=4]{\tick}\%},
    scaled x ticks=false,
    scaled y ticks=false,
    grid=both,
    grid style={gray!20},
    axis lines=left,
    axis line style={gray!60},
    tick label style={font=\small, gray!70!black},
    label style={font=\small},
    title style={font=\small\bfseries},
]
\addplot[
    color=gray!70!black,
    line width=1pt,
    mark=none,
] coordinates {
    (500000, 0.03140)
    (2738176, 0.02903)
    (3811328, 0.02839)
    (4613120, 0.02827)
    (5292544, 0.02808)
    (5893120, 0.02803)
    (6444032, 0.02806)
    (6956032, 0.02816)
    (7869440, 0.02815)
    (8692224, 0.02839)
    (9441792, 0.02836)
    (10000000, 0.02833)
    (10510336, 0.02864)
    (11092992, 0.02850)
    (11645440, 0.02847)
    (12175360, 0.02852)
    (12684288, 0.02852)
    (13253632, 0.02826)
    (13798912, 0.02839)
    (14325760, 0.02828)
    (14834176, 0.02837)
    (15394304, 0.02847)
    (15936000, 0.02829)
    (16459264, 0.02834)
    (16967168, 0.02829)
    (17521664, 0.02834)
    (18060288, 0.02838)
    (18583552, 0.02841)
    (19092992, 0.02845)
    (19642880, 0.02849)
    (20179456, 0.02846)
    (20700672, 0.02848)
    (21210624, 0.02854)
    (21757440, 0.02858)
    (22290432, 0.02852)
    (22812160, 0.02845)
    (23323648, 0.02841)
    (23867392, 0.02834)
    (24398336, 0.02824)
    (24919040, 0.02815)
    (25430016, 0.02811)
    (25931264, 0.02816)
    (26463232, 0.02821)
    (26983424, 0.02816)
    (27494400, 0.02814)
    (27997184, 0.02817)
    (28528640, 0.02818)
    (29049344, 0.02809)
    (29565440, 0.02805)
    (30070272, 0.02805)
    (30574592, 0.02804)
    (31106048, 0.02811)
    (31617024, 0.02811)
    (32122880, 0.02805)
    (32623616, 0.02805)
    (33150976, 0.02807)
    (33675264, 0.02808)
    (34188288, 0.02808)
    (34695168, 0.02807)
    (35225088, 0.02803)
    (35750400, 0.02803)
    (36264960, 0.02803)
    (36774400, 0.02801)
    (37277184, 0.02798)
    (37803520, 0.02789)
    (38321664, 0.02786)
    (38834176, 0.02789)
    (39339008, 0.02787)
    (39866368, 0.02790)
    (40481792, 0.02789)
    (40995328, 0.02788)
    (41503232, 0.02789)
    (42492928, 0.02786)
    (43003392, 0.02784)
    (43510272, 0.02786)
    (44011008, 0.02781)
    (45019648, 0.02781)
    (46030848, 0.02778)
    (47048192, 0.02774)
    (48065536, 0.02777)
    (49081344, 0.02768)
    (50000000, 0.02769)
    (50506240, 0.02770)
    (51017216, 0.02765)
    (51520000, 0.02764)
    (52022784, 0.02759)
    (52540416, 0.02758)
    (53048320, 0.02757)
    (53557248, 0.02756)
    (54057984, 0.02756)
    (54576640, 0.02754)
    (55091712, 0.02750)
    (55597568, 0.02751)
    (56103424, 0.02749)
    (56605184, 0.02744)
    (57108992, 0.02749)
    (57610240, 0.02748)
    (58122752, 0.02749)
    (58630144, 0.02752)
    (59139072, 0.02752)
    (59655680, 0.02748)
    (60171776, 0.02750)
    (60678144, 0.02748)
    (61182976, 0.02745)
    (61700608, 0.02745)
    (62215680, 0.02745)
    (62723584, 0.02748)
    (63227904, 0.02747)
    (63731200, 0.02747)
    (64233472, 0.02744)
    (64743424, 0.02744)
    (65250304, 0.02741)
    (65757184, 0.02743)
    (66258432, 0.02742)
    (66769920, 0.02740)
    (67276800, 0.02742)
    (67784192, 0.02743)
    (68288512, 0.02745)
    (68803584, 0.02746)
    (69311488, 0.02748)
    (69818880, 0.02746)
    (70323712, 0.02745)
    (70827520, 0.02745)
    (71342592, 0.02741)
    (71849472, 0.02741)
    (72356352, 0.02739)
    (72861696, 0.02739)
    (73363456, 0.02741)
    (73868800, 0.02741)
    (74373120, 0.02742)
    (74877952, 0.02742)
    (75381760, 0.02744)
    (75884544, 0.02741)
    (76394496, 0.02742)
    (76895232, 0.02743)
    (77396992, 0.02742)
    (77898240, 0.02742)
    (78398976, 0.02742)
    (78911488, 0.02740)
    (79413760, 0.02737)
    (80159232, 0.02738)
    (81425920, 0.02740)
    (82426368, 0.02739)
    (83438080, 0.02739)
    (84447744, 0.02737)
    (85450752, 0.02740)
    (86461440, 0.02741)
    (87472640, 0.02737)
    (88475136, 0.02734)
    (89477632, 0.02733)
    (90264064, 0.02732)
    (91268608, 0.02731)
    (92273152, 0.02731)
    (93284352, 0.02730)
    (94295040, 0.02730)
    (95295488, 0.02729)
    (96306688, 0.02727)
    (97312256, 0.02728)
    (98055168, 0.02729)
    (98794496, 0.02728)
    (99696128, 0.02729)
    (100000000, 0.02730)
    (101000000, 0.02728)
    (102000000, 0.02731)
    (103000000, 0.02731)
    (104000000, 0.02732)
    (105000000, 0.02733)
    (106000000, 0.02734)
    (107000000, 0.02734)
    (108000000, 0.02734)
    (109000000, 0.02734)
    (110000000, 0.02735)
    (111000000, 0.02734)
    (112000000, 0.02732)
    (113000000, 0.02729)
    (114000000, 0.02731)
    (115000000, 0.02731)
    (116000000, 0.02731)
    (117000000, 0.02730)
    (118000000, 0.02731)
    (119000000, 0.02732)
    (120000000, 0.02732)
    (121000000, 0.02731)
    (122000000, 0.02731)
    (123000000, 0.02732)
    (124000000, 0.02734)
    (125000000, 0.02734)
    (126000000, 0.02731)
    (127000000, 0.02730)
    (128000000, 0.02728)
    (129000000, 0.02728)
    (130000000, 0.02729)
    (131000000, 0.02729)
    (132000000, 0.02728)
    (133000000, 0.02727)
    (134000000, 0.02724)
    (135000000, 0.02722)
    (136000000, 0.02721)
    (137000000, 0.02721)
    (138000000, 0.02721)
    (139000000, 0.02722)
    (140000000, 0.02721)
    (141000000, 0.02721)
    (142000000, 0.02720)
    (143000000, 0.02721)
    (144000000, 0.02722)
    (145000000, 0.02723)
    (146000000, 0.02722)
    (147000000, 0.02722)
    (148000000, 0.02723)
    (149000000, 0.02723)
    (150000000, 0.02723)
    (151000000, 0.02725)
    (152000000, 0.02724)
    (153000000, 0.02723)
    (154000000, 0.02723)
    (155000000, 0.02723)
    (156000000, 0.02720)
    (157000000, 0.02721)
    (158000000, 0.02722)
    (159000000, 0.02722)
    (160000000, 0.02723)
    (161000000, 0.02724)
    (162000000, 0.02723)
    (163000000, 0.02724)
    (164000000, 0.02723)
    (165000000, 0.02723)
    (166000000, 0.02723)
    (167000000, 0.02722)
    (168000000, 0.02721)
    (169000000, 0.02721)
    (170000000, 0.02722)
    (171000000, 0.02721)
    (172000000, 0.02720)
    (173000000, 0.02721)
    (174000000, 0.02719)
    (175000000, 0.02720)
    (176000000, 0.02721)
    (177000000, 0.02719)
    (178000000, 0.02721)
    (179000000, 0.02721)
    (180000000, 0.02721)
    (181000000, 0.02720)
    (182000000, 0.02720)
    (183000000, 0.02719)
    (184000000, 0.02719)
    (185000000, 0.02718)
    (186000000, 0.02718)
    (187000000, 0.02718)
    (188000000, 0.02718)
    (189000000, 0.02718)
    (190000000, 0.02718)
    (191000000, 0.02718)
    (192000000, 0.02719)
    (193000000, 0.02721)
    (194000000, 0.02720)
    (195000000, 0.02718)
    (196000000, 0.02719)
    (197000000, 0.02720)
    (198000000, 0.02721)
    (199000000, 0.02720)
    (200000000, 0.02721)
};
\addplot[
    color=gray!70!black,
    densely dashed,
    line width=0.6pt,
    mark=none,
] coordinates {(0, 0.0272) (204000000, 0.0272)};
\end{axis}

\end{tikzpicture}
\caption{\label{fig.proporcao}Proportion of values of $n$ with $\nimpt{n} \neq CF_{n+3}$, that is, divergences up to $n$ divided by $n$, for $n$ up to $200$M. The proportion at $n$ = $200$M is $0.0272\%$ ($54{,}412$ divergences).}
\end{figure}

\section{Hull game in convexity $\ell_k$ with $k$ odd}
 
The previous sections showed how to compute the nimbers $\nimp{n}$ and $\nim_{\ell_k}(C_n)$ for every $k,n\ge2$. In this section, we show that the nimber sequences of these games are periodic when $k$ is odd.
We start by proving the periodicity of $\nimpb{n}$ and $\nimpc{n}$.
 
\begin{lemma}\label{lem:k-odd}
Let $k\ge3$ odd and $n\ge1$. Then $\nimpb{n}=(n-1)\%(k+1)$ and
\[
\nimpc{n}=
\begin{cases}
0, &\mbox{if $n<k+1$},\\
n\%(k+1), &\mbox{if $n\ge k+1$},
\end{cases}
\]
where $\%$ is the remainder operator.
\end{lemma}
 
\begin{proof}
First consider $P''_n$.
Theorem~\ref{thm-nimber-Pn} and Lemma~\ref{lem:phaseA} obtained these values for $n$ from $1$ to $3k+2$. So, let $n\ge 3k+3$ and $r=n\%(k+1)$.
Notice that, in the recurrence of $\nimpc{n}$ given by Theorem~\ref{thm-nimber-Pn}, the $\mex$ function contains the values $\nimpc{2}\xor\nimpc{n-1},\ldots,\nimpc{k+1}\xor\nimpc{n-k}$. Since the values $\nimpc{2},\ldots,\nimpc{k+1}$ are $0$ from Theorem~\ref{thm-nimber-Pn}, we have by induction that the $\mex$ function contains the values $\{0,\ldots,k\}\setminus\{r\}$, that is, except $r=n\%(k+1)$.
 
Moreover, the $\mex$ function also includes the values $\nimpc{i}\xor\nimpc{j}$ such that $i+j=n+1$ and $i,j\ge k+2$. Since $k+1$ is even, we have by induction that $\nimpc{i}$ is even if and only if $i$ is even. Similarly for $j$. Moreover $i\xor j$ is even if and only if $i\%2=j\%2$, which is true if and only if $n$ is odd, since $i+j=n+1$. Therefore, $\nimpc{i}\xor\nimpc{j}$ is even if and only if $n$ is odd. Consequently, $\nimpc{i}\xor\nimpc{j}$ can never be equal to $r=n\%(k+1)$, which has the same parity of $n$, and the $\mex$ function does not include $r$. Since all possible values are of this type, we have from the Sprague-Grundy Theorem~\ref{thm-sprague} that $\nimpc{n}=r$.
 
Now consider $P'_n$. From Theorem~\ref{thm:shift},
\[
\nimpb{n}\ =\ \nimpc{n+k}\ =\ (n+k)\%(k+1)\ =\ (n-1)\%(k+1).
\]
\end{proof}
 
With this, we have the following nimber sequences for $\nimpb{n}$ and $\nimpc{n}$ with $k$ odd:
 
\[
\nimpb{n}:\ \ \
\underset{\scriptscriptstyle 1}{0},\dots,\underset{\scriptscriptstyle k+1}{k},\;
\underset{\scriptscriptstyle k+2}{0},\underset{\scriptscriptstyle k+3}{1},\dots,\underset{\scriptscriptstyle 2k+1}{k-1},\underset{\scriptscriptstyle 2k+2}{k},\;
\underset{\scriptscriptstyle 2k+3}{0},\dots,\underset{\scriptscriptstyle 3k+2}{k-1},\underset{\scriptscriptstyle 3k+3}{k}.
\]
 
\[
\nimpc{n}:\ \ \
\underset{\scriptscriptstyle 1}{0},\dots,\underset{\scriptscriptstyle k+1}{0},\;
\underset{\scriptscriptstyle k+2}{1},\underset{\scriptscriptstyle k+3}{2},\dots,\underset{\scriptscriptstyle 2k+1}{k},\underset{\scriptscriptstyle 2k+2}{0},\;
\underset{\scriptscriptstyle 2k+3}{1},\dots,\underset{\scriptscriptstyle 3k+2}{k},\underset{\scriptscriptstyle 3k+3}{0}.
\]

Now we determine the winner in $P_n$ when $k$ is odd.
 
\begin{theorem}\label{thm:k-odd}
Let $k\ge3$ odd and $n\ge1$. Then
\[
\nimp{n}\ =\
\begin{cases}
0, &\mbox{if $n$ is even},\\
1, &\mbox{if $n$ is odd}.
\end{cases}
\]
Therefore, Alice wins the game $\CHG_{\ell_k}$ if and only if $n$ is odd.
\end{theorem}
 
\begin{proof}
Notice that, in the recurrence of $\nimp{n}$ given by Theorem~\ref{thm-nimber-Pn}, the $\mex$ function contains the values $\nimpb{i}\xor\nimpb{j}$ such that $i+j=n+1$ and $i,j\ge1$. Since $k+1$ is even, we have by Lemma~\ref{lem:k-odd} that $\nimpb{i}\xor\nimpb{j}$ has the same parity of $(i-1)\xor(j-1)$.
 
First assume that $n$ is even. Then $i$ is even if and only if $j$ is odd (since $i+j=n+1$ is odd) and consequently $\nimpb{i}\xor\nimpb{j}$ is always odd. Thus, the $\mex$ function does not include $0$ and then $\nimp{n}=0$.
 
Now assume that $n$ is odd. Then $i$ and $j$ has the same parity (since $i+j=n+1$ is even) and consequently $\nimpb{i}\xor\nimpb{j}$ is always even. Moreover, taking $i=j=(n+1)/2$, we have the value $0$ in the $\mex$ function and no odd value. Then, $\nimp{n}=1$.
\end{proof}

\begin{corollary}\label{cor:k-odd}
Let $k\ge3$ odd and $n\ge3$. Then
\[
\nim_{\ell_k}(C_n)\ =\ 
\begin{cases}
1, &\mbox{if $n=3$ or $n=\alpha\cdot(k+1)-1$ with $\alpha\ge2$},\\
0, &\mbox{otherwise}.
\end{cases}
\]
\end{corollary}
 
\begin{proof}
It follows directly from Corollary~\ref{cor-nimber-Cn} and Lemma~\ref{lem:k-odd}.
\end{proof}
 
\section{Hull game in convexity $\ell_k$ with $k$ even}
 
An interesting empirical observation is that, across extensive computational experiments, the nimber sequences $\nimp{n}$ for even $k$ never appeared to become periodic except when $k=2^h-4$ for $h\ge3$, that is, $k=4,12,28,60,124,252,\dots$
The rest of this section focuses exclusively on this family of values of $k=2^h-4$ for an integer $h\ge 3$.
We prove that the nimber sequences $\nimpc{n}$, $\nimpb{n}$ and $\nimp{n}$ for $n\ge1$ are periodic:
\begin{align*}
\nimpc{n}:&\ \ 0^{k+1}, 1, \ldots,k, 0, 1, \ldots, k+1, 1, \ldots, k+1, 1, \ldots, k-1, 0, k+1, [1,\ldots, k+1]^*\\
\nimpb{n}:&\ \ \ \ \ \ 0, 1, \ldots,k, 0, 1, \ldots, k+1, 1, \ldots, k+1, 1, \ldots, k-1, 0, k+1, [1,\ldots, k+1]^*\\
\nimp{n}:&\ \ \ \ \ \ [1,0]^{k/2}, 1, 2, 2, B^b, 1, 1, 1, 2, 2, B^b, 1, 1, 3, 0, 2, [B^b, C]^*,
\end{align*}
where $b=(k-4)/8$, $B=11331122$ and $C=11323$.
This will imply that, in the game $\CHG_{\ell_k}$ with $n>k$, Alice loses in the path $P_n$ only for $n=3k+4$ and she wins in the cycle $C_n$ only for $n\in\{2k+1,5k+3\}$.
 
\vspace{0.5cm}
The proof below is a long case analysis, since the periodic structures are non-trivial and depend on $k$. We start by proving in Theorem~\ref{thm:family-closed-form} and Corollary~\ref{cor:Pn2} the periodic structure of the sequences $\nimpc{n}$ and $\nimpb{n}$. Later we prove Lemmas~\ref{lem:nimpb-small} to~\ref{lem:tail} in order to prove Theorem~\ref{thm:Pn} with the periodic structure of $\nimp{n}$.
 
We start by proving the following fact, valid for any $k$ even.
 
\begin{lemma}\label{lem:phaseBwrap}
For even $k$, $\nimpc{3k+3} = k+1$.
\end{lemma}
 
\begin{proof}
Let $n=3k+3$, which is odd, since $k$ is even.
By Theorem~\ref{thm-nimber-Pn}, the $\mex$ function for $\nimpc{n}$ includes every value $\nimpc{i}\xor\nimpc{j}$ such that $i+j=3k+4$ and $i\ge j\ge2$. Taking $i=j=(3k+4)/2$ gives the value $0$.
For $j=2,\dots,k+1$, $i=3k+4-j$ ranges over $\{2k+3,\dots,3k+2\}$, inside the range of Lemma~\ref{lem:phaseA}, where $\nimpc{j}=0$; combined with the value $0$ obtained above, this gives the set $\{0,1,\dots,k\}$, missing only $k+1$.
 
As a further check, the $\mex$ function also includes the value $\nimpc{2k+2}\xor\nimpc{k+2}=0\xor 1=1\ne k+1$, from Lemma~\ref{lem:phaseA}.
For the remaining pairs $(2k+1,k+3), (2k,k+4),\ldots$ with $2k+1\ge i\ge j$ and $i+j=3k+4$, the residues $i\%(k+1)$ and $j\%(k+1)$ have the same parity, so $\nimpc{i}\xor\nimpc{j}$ is even (from Lemma~\ref{lem:phaseA}), and can never equal $k+1$, which is odd.
Hence $\nimpc{3k+3}=\mex\{0,\dots,k\}=k+1$.
\end{proof}
 
Therefore, with $k$ even, the nimber sequence of $\nimpc{n}$ is (up to $n=3k+3$)
\[
\nimpc{n}:\ \ \ \underset{\scriptscriptstyle 1}{0},\dots,\underset{\scriptscriptstyle k+1}{0},\;
\underset{\scriptscriptstyle k+2}{1},\underset{\scriptscriptstyle k+3}{2},\dots,\underset{\scriptscriptstyle 2k+1}{k},\underset{\scriptscriptstyle 2k+2}{0},\;
\underset{\scriptscriptstyle 2k+3}{1},\dots,\underset{\scriptscriptstyle 3k+2}{k},\underset{\scriptscriptstyle 3k+3}{k+1}.
\]
 
 
\subsection{Nimbers of $P'_n$ and $P''_n$ in $\CHG_{\ell_k}$ with $k=2^h-4$ and $h\ge3$}
 
Throughout this subsection we use repeatedly the elementary identity
\[
x+y\ =\ (x\xor y) + 2\,(x\mathbin{\&}y), \tag{$\ast$}
\]
valid for all nonnegative integers $x,y$, where $\mathbin{\&}$ denotes the bitwise AND. In particular, $x\xor y$ always has the same parity as $x+y$.
 
Let $\rho = \dfrac{k+2}{2} = 2^{h-1}-1$, a \emph{binary repunit}: in base $2$, $\rho$ consists of $h-1$ consecutive $1$-bits (bits $0$ through $h-2$). This property of $\rho$ is what makes the family $k=2^h-4$ special: the obstructions encountered below never materialize precisely because the low bits of $\rho$ are entirely full.
 
By Lemma~\ref{lem:phaseA} (valid for $1\le n\le 3k+2$) together with Lemma~\ref{lem:phaseBwrap} (valid for $n=3k+3$), the sequence $\nimpc{n}$ is completely known up to $n=3k+3$: it equals $n\%(k+1)$ throughout, except that the multiple of $k+1$ at $n=3(k+1)$ takes the value $k+1$ instead of $0$. The next theorem shows that this same rule continues to hold for every larger $n$, with a single further isolated exception.
 
For every integer $n$, let $\sigma(n)$ denote the unique element of $\{1,\ldots,k+1\}$ congruent to $n$ modulo $k+1$; that is, $\sigma(n)=n\%(k+1)$ if $n\%(k+1)\neq0$, and $\sigma(n)=k+1$ otherwise.
 
\begin{theorem}
\label{thm:family-closed-form}
Let $k=2^h-4$ with $h\ge3$. For every $n\ge 3k+4$,
\[
\nimpc{n} =
\begin{cases}
0, & \mbox{if}\ n = 5k+4,\\
k+1, & \mbox{if}\ n\%(k+1)=0,\\
n\%(k+1), & \text{otherwise,}
\end{cases}
\]
that is, $\nimpc{n}=\sigma(n)$ for every $n\ge3k+4$ with $n\neq5k+4$, and $\nimpc{5k+4}=0$.
\end{theorem}
 
\begin{proof}
Let $p''(n)=\nimpc{n}$ and $E=\{2k+2,\,5k+4\}$ be the exceptional set. Recall that $k\ge4$ is even, so $k+1\ge5$ is odd.
By Lemmas~\ref{lem:phaseA} and~\ref{lem:phaseBwrap}, $p''(n)=0$ for $1\le n\le k+1$, $p''(2k+2)=0$, and $p''(n)=\sigma(n)$ for every other $n\in\{k+2,\ldots,3k+3\}$.
We prove the theorem by induction on $n\ge3k+4$.
By Theorem~\ref{thm-nimber-Pn}, $p''(n)=\mex S(n)$, where
\[
S(n)=\bigl\{p''(i)\xor p''(j)\ :\ i+j=n+1,\ \ \ 2\le i\le j\le n-1\bigr\}.
\]
We split the pairs $\{i,j\}$ (with $i\le j$) into three types:
\begin{itemize}[nosep]
\item[($T_1$)] $i\le k+1$. Then $p''(i)=0$ and the term is $p''(j)$, where $j$ ranges over the $k$ consecutive integers $n-k,\ldots,n-1$, all at least $2k+4$.
\item[($T_2$)] $i,j\ge k+2$ and $i,j\notin E$. The term is $x\xor y$ with $x=\sigma(i)$, $y=\sigma(j)\in\{1,\ldots,k+1\}$, where $x+y\equiv i+j=n+1\pmod{k+1}$ and $2\le x+y\le2k+2$.
\item[($T_3$)] $i,j\ge k+2$ and $\{i,j\}\cap E\neq\emptyset$.
\end{itemize}
The hypothesis $k=2^h-4$ is used only in the following claim, applied to type $T_2$.
 
\medskip\noindent\textbf{Claim.} Let $x,y\in\{1,\ldots,k+1\}$ and $\tau\ge1$ with $x+y=k+2+\tau$. Then $x\xor y\ne\tau$.
 
\noindent\emph{Proof of the claim.} Suppose $x\xor y=\tau$. By $(\ast)$,
$x\mathbin{\&}y=\bigl((k+2+\tau)-\tau\bigr)/2=\rho=2^{h-1}-1$, so bits
$0,\ldots,h-2$ are set in both $x$ and $y$. Since $x,y\le k+1<2^h$, the only
other bit available is bit $h-1$, hence $x,y\in\{2^{h-1}-1,\,2^h-1\}$; as
$2^h-1=k+3>k+1$, we get $x=y=\rho$. Then $x\xor y=0\ne\tau$, a contradiction.
\hfill$\diamond$
 
\medskip
\noindent\textbf{Step 1: the case $n=5k+4$.} Here $n+1=5(k+1)\equiv0\pmod{k+1}$, and $n+1$ is odd, so there is no self-pair $i=j$. We show $0\notin S(n)$.
\begin{itemize}[nosep]
\item Type $T_1$: $j\in\{4k+4,\ldots,5k+3\}$, which is disjoint from $E$, so $p''(j)=\sigma(j)\ge1$.
\item Type $T_2$: $x\xor y=0$ implies $x=y$, and $x+y\equiv0\pmod{k+1}$ with $2\le x+y\le2k+2$ implies $2x\in\{k+1,2k+2\}$. The first case is impossible because $k+1$ is odd. The second gives $x=y=k+1$, so $i,j$ are multiples of $k+1$ outside $E$ and at least $k+2$, hence $i,j\ge3k+3$ and $i+j\ge6k+6>n+1$, a contradiction.
\item Type $T_3$: only $2k+2\in E$ is smaller than $n$. Its partner is $3k+3$, and the term is $0\xor p''(3k+3)=k+1\ne0$.
\end{itemize}
Hence $p''(5k+4)=\mex S(5k+4)=0$.
 
\medskip
\noindent\textbf{Step 2: the case $n\neq5k+4$.} Let $\tau:=\sigma(n)\in\{1,\ldots,k+1\}$. Then $n+1\equiv\tau+1\pmod{k+1}$. We show that $\tau\notin S(n)$ and $\{0,\ldots,\tau-1\}\subseteq S(n)$, so that $p''(n)=\tau$.
 
\smallskip\noindent\emph{(i) $\tau\notin S(n)$.}
\begin{itemize}[nosep]
\item Type $T_1$: the $k$ consecutive integers $n-k,\ldots,n-1$ cover every residue modulo $k+1$ except that of $n$. So $\sigma(j)\ne\tau$, and $p''(j)\in\{0,\sigma(j)\}$ differs from $\tau\ge1$.
\item Type $T_2$: $x+y\equiv\tau+1\pmod{k+1}$ and $2\le x+y\le2k+2$, with $2\le\tau+1\le k+2$. Hence $x+y=\tau+1$ or $x+y=\tau+k+2$. In the first case, by $(\ast)$, $x\xor y$ has the parity of $\tau+1$, so $x\xor y\ne\tau$. In the second case, $x\xor y\ne\tau$ by the claim.
\item Type $T_3$ with $2k+2$: the partner is $j=n-2k-1\ge k+3$, with $j\equiv n+1\pmod{k+1}$. If $j\in E$, the term is $0\ne\tau$. Otherwise the term is $\sigma(j)=\sigma(n+1)\in\{\tau+1,\,1\}$, where the value $1$ occurs only if $\tau=k+1$. In both cases the term differs from $\tau$.
\item Type $T_3$ with $5k+4$ (only when $n\ge5k+5$): the partner is $j=n-5k-3$. If $j\le k+1$, the pair is of type $T_1$. If $j\in E$, the term is $0\ne\tau$. Otherwise, since $5k+3\equiv-2\pmod{k+1}$, the term is $\sigma(j)\equiv j=n-(5k+3)\equiv n+2\equiv\tau+2\pmod{k+1}$, which cannot be
$\tau$, since $k+1\ge5$.
\end{itemize}
 
\smallskip\noindent\emph{(ii) $\{1,\ldots,\tau-1\}\subseteq S(n)$.} By (i), the integers $j\in\{n-k,\ldots,n-1\}$ of type $T_1$ satisfy $\{\sigma(j)\}=\{1,\ldots,k+1\}\setminus\{\tau\}$. None of these $j$ equals $2k+2$ (as $j\ge2k+4$), so $p''(j)=\sigma(j)$ for all of them, except when $5k+4\in\{n-k,\ldots,n-1\}$. In that case $p''(5k+4)=0$ replaces the value $\sigma(5k+4)=k$. This loses a needed value only if $k\le\tau-1$, that is, $\tau=k+1$, which together with $5k+5\le n\le6k+4$ means $n=5k+5$. In that case, the pair $\{k+2,\,4k+4\}$ (sum $5k+6=n+1$, both indices outside $E$) contributes $p''(k+2)\xor p''(4k+4)=1\xor(k+1)=k$, because $k$ is even.
 
\smallskip\noindent\emph{(iii) $0\in S(n)$.} If $n$ is odd, the self-pair $i=j=(n+1)/2$ gives $0$. Let $n$ be even. 
Let $s=\sigma(n+1)/2$ if $\sigma(n+1)$ is even, or $s=(\sigma(n+1)+k+1)/2$, otherwise.
In both cases, $s\in\{1,\ldots,k+1\}$ and $2s\equiv \sigma(n+1)\equiv n+1\pmod{k+1}$.
Write $n+1-2s=N\cdot(k+1)$. Here $N(k+1)$ is odd because $n+1$ is odd, so $N$ is odd. Also $N(k+1)\ge(3k+5)-(2k+2)=k+3>k+1$, so $N\ge3$.
Take $i=2(k+1)+s$ and $j=(N-2)(k+1)+s$.
Then $i+j=n+1$ and $i,j\ge k+2$. Moreover $i\notin E$: $i=2k+2$ would need $s=0$, and $i=5k+4$ would need $s=3k+2$. Also $j\notin E$: $j=2k+2$ would need $N=3$ and $s=k+1$, that is, $n=5k+4$; $j=5k+4$ would need $N=6$, which is even. By induction, $p''(i)=p''(j)=\sigma(i)=\sigma(j)=s$, so $p''(i)\xor p''(j)=0\in S(n)$.
 
\end{proof}

\begin{corollary}\label{cor:Pn2}
Let $k=2^h-4$ with $h\ge3$.
Then the nimber sequences $\nimpb{n}$ and $\nimpc{n}$ are periodic with period $k+1$ and their transients have length $4k+4$ and $5k+4$, respectively. Moreover, $\nimpb{n}=0$ implies $n=1$, $n=k+2$ or $n=4k+4$, and $\nimpc{n}=0$ implies $n\le k+1$, $n=2k+2$ or $n=5k+4$.
\end{corollary}
 
\begin{proof}
Directly from Theorem~\ref{thm:family-closed-form} and the fact that $\nimpb{n}=\nimpc{n+k}$ of Theorem~\ref{thm:shift}
\end{proof}
 
\begin{corollary}\label{cor:Cn2}
Let $k=2^h-4$ with $h\ge3$ and $n>k$. Then $\nim_{\ell_k}(C_n)=1$ if and only if $n\in\{2k+1,5k+3\}$; otherwise $\nim_{\ell_k}(C_n)=0$. In particular, Alice wins $\CHG_{\ell_k}$ on $C_n$ if and only if $n\in\{2k+1,5k+3\}$.
\end{corollary}
 
\begin{proof}
Directly from Corollary~\ref{cor-nimber-Cn}, Lemma~\ref{lem:phaseA} and Theorem~\ref{thm:family-closed-form}.
\end{proof}
 
 
\subsection{Nimbers of $P_n$ in $\CHG_{\ell_k}$ with $k=2^h-4$ and $h\ge3$}
 
The sequence $\nim_{\ell_k}(P_n)$, $n\ge1$, decomposes into six initial blocks (the transient) followed by an infinite periodic tail; Table~\ref{tab:blocks} summarizes this decomposition -- the range of $n$ covered by each block, its length, the value(s) of $\nimp{n}$ it asserts, and the lemma in which it is proved -- as a roadmap for the whole subsection.
 
\begin{table}[ht]
\centering
\begin{tabular}{c|c|c|c|c}
Block & Range of $n$ & Length & $\nim_{\ell_k}(P_n)$ & Proved in\\
\hline
1 & $1,\ldots,k$ & $k$ & $[1,0]^{k/2}$ & Lemma~\ref{lem:block1}\\
2 & $k+1,\ldots,k+3$ & $3$ & $1,\,2,\,2$ & Lemma~\ref{lem:block2}\\
3 & $k+4,\ldots,2k-1$ & $8b=k-4$ & $B^b$ & Lemma~\ref{lem:block3}\\
4 & $2k,\ldots,2k+4$ & $5$ & $1,\,1,\,1,\,2,\,2$ & Lemma~\ref{lem:block4}\\
5 & $2k+5,\ldots,3k$ & $8b=k-4$ & $B^b$ & Lemma~\ref{lem:block5}\\
6 & $3k+1,\ldots,3k+5$ & $5$ & $1,\,1,\,3,\,0,\,2$ & Lemma~\ref{lem:block6}\\
\hline
Tail & $n\ge3k+6$ & period $k+1$ & $[B^b,C]^*$ & Lemma~\ref{lem:tail}\\
\end{tabular}
\caption{Decomposition of the sequence $\nim_{\ell_k}(P_n)$, $n\ge1$, into the six blocks of the transient (of total length $3k+5$) and the periodic tail, with $B=11331122$, $C=11323$ and $b=(k-4)/8$ as in Theorem~\ref{thm:Pn}. Blocks~3 and~5 each have length $8b=k-4$.}
\label{tab:blocks}
\end{table}
 
We first record some direct consequences of Theorems~\ref{thm:shift} and~\ref{thm:family-closed-form}, giving $\nimpb{n}$ explicitly on the small range that will drive the whole computation of $\nimp{n}$ below. Recall that $\sigma(n)=n\%(k+1)$ if $n\%(k+1)\neq0$, and $\sigma(n)=k+1$ otherwise.
 
\begin{lemma}\label{lem:nimpb-small}
Let $k=2^h-4$ with $h\ge3$. For every $n\ge1$,
\[
\nimpb{n}=
\begin{cases}
0, & \mbox{if}\ n\in\{1,\,k+2,\,4k+4\},\\
\sigma(n-1), & \mbox{otherwise.}
\end{cases}
\]
In particular,
\[
\nimpb{n} =
\begin{cases}
n-1, & \mbox{if}\ \ 1\le n\le k+1,\\
n-(k+2), & \mbox{if}\ \ k+2\le n\le 2k+3,\\
n-(2k+3), & \mbox{if}\ \ 2k+4\le n\le 3k+4,\\
n-(3k+4), & \mbox{if}\ \ 3k+5\le n\le 4k+3.
\end{cases}
\]
\end{lemma}
 
\begin{proof}
By Lemmas~\ref{lem:phaseA} and~\ref{lem:phaseBwrap} and Theorem~\ref{thm:family-closed-form}, for every $m\ge1$ we have $\nimpc{m}=0$ if $m\le k+1$ or $m\in\{2k+2,5k+4\}$, and $\nimpc{m}=\sigma(m)$ otherwise. By Theorem~\ref{thm:shift}, $\nimpb{n}=\nimpc{m}$ with $m=n+k$. Now $m\le k+1$ iff $n=1$, $m=2k+2$ iff $n=k+2$, and $m=5k+4$ iff $n=4k+4$. Moreover, $\sigma(n+k)=\sigma(n-1)$, since $k\equiv-1\pmod{k+1}$.
 
For the 2nd formula, note that if $q(k+1)+1\le n-1\le(q+1)(k+1)$, then $\sigma(n-1)=n-1-q(k+1)$. Taking $q=0,1,2,3$ gives the four expressions $n-1$, $n-(k+2)$, $n-(2k+3)$ and $n-(3k+4)$. The three exceptional points agree with them: $n=1$ and $n=k+2$ are the first points of the first two ranges, with value $0$, and $n=4k+4$ lies outside the four ranges.
\end{proof}
 
The recurrence of Theorem~\ref{thm-nimber-Pn} expresses $\nimp{n}$ directly in terms of $\nimpb{i}$; since the term $\nimpb{i}\xor\nimpb{n-i+1}$ is unchanged under $i\leftrightarrow n-i+1$,
\[
\nimp{n}\ =\ \mex\bigl\{\nimpb{i}\xor\nimpb{n+1-i} : i=1,\ldots,n\bigr\}.
\]
We use this recurrence throughout. We will use the following general fact about fixed-sum $\xor$, which lets us decide, directly from $T\%8$, exactly which of the values $0,1,2,3$ a family $\{x\xor(T-x):x=0,\ldots,T\}$ contains.
 
\begin{lemma}\label{lem:xor-realize}
Let $T,v\ge0$. There exists $x\in\{0,\ldots,T\}$ with $x\xor(T-x)=v$ if and only if
\[
T\ge v,\qquad T\equiv v\pmod2,\qquad \text{and}\qquad v\mathbin{\&}w=0,\ \text{where } w=\tfrac{T-v}{2}.
\]
Moreover, whenever these three conditions hold, $x=w$ is such a witness.
\end{lemma}
 
\begin{proof}
($\Leftarrow$) Suppose the three conditions hold, and set $x=w$, $y=v+w$. Since $v\mathbin{\&}w=0$, the numbers $v$ and $w$ have disjoint sets of set bits, so the sum $y=v+w$ has no carries and its set bits are exactly those of $v$ together with those of $w$. Hence $x\mathbin{\&}y=w\mathbin{\&}y=w$. By $(\ast)$, $x\xor y=(x+y)-2(x\mathbin{\&}y)=(v+2w)-2w=v$, and $x+y=v+2w=v+(T-v)=T$, so $y=T-x$. Also $0\le x=w\le T$ since $0\le v\le T$. Thus $x=w\in\{0,\ldots,T\}$ is a witness.
 
($\Rightarrow$) Suppose $v=x\xor(T-x)$ for some $x\in\{0,\ldots,T\}$; set $y=T-x$. By $(\ast)$, $T=x+y=v+2(x\mathbin{\&}y)$, so $T\ge v$, $T\equiv v\pmod 2$, and $x\mathbin{\&}y=(T-v)/2=w$. Since $x\xor y=v$, every set bit of $v$ is a bit where exactly one of $x,y$ is set, hence not a bit of $x\mathbin{\&}y=w$. Thus $v\mathbin{\&}w=0$.
\end{proof}
 
Since Lemma~\ref{lem:xor-realize}'s three conditions, for a fixed target $v\in\{0,1,2,3\}$, only ever inspect $T$ through $T\ge v$ (automatic once $T\ge3$) and through the two low bits of $w=(T-v)/2$ (because $v\le3$ has no bits above bit~$1$), the set
\[
\Phi(T) := \{\, v\in\{0,1,2,3\} : v=x\xor(T-x) \text{ for some } x\in\{0,\ldots,T\}\,\}
\]
depends, for every $T\ge3$, only on $T\%8$.
Direct case-checking of $T=0,\ldots,10$ via Lemma~\ref{lem:xor-realize} gives Table~\ref{tab:Phi} for these values of $T$; since $\{3,\ldots,10\}$ contains every residue modulo $8$ and $\Phi(T)$ depends only on $T\%8$ for $T\ge3$, Table~\ref{tab:Phi} holds for every $T\ge0$.
We will often use the following simple fact: $\mex S=\mex\bigl(S\cap\{0,1,2,3\}\bigr)$ whenever $\{0,1,2,3\}\not\subseteq S$.
In particular, if $S\cap\{0,1,2,3\}=\Phi(T)$, then $\mex S=\mex\Phi(T)$, since no set $\Phi(T)$ in Table~\ref{tab:Phi} equals $\{0,1,2,3\}$. 
 
\begin{table}[ht]
\centering
\begin{tabular}{c|c|c|c|c|c|c|c|c}
$T\%8$ & 0 & 1 & 2 & 3 & 4 & 5 & 6 & 7\\\hline
$\Phi(T)$ & $\{0\}$ & $\{1\}$ & $\{0,2\}$ & $\{3\}$ & $\{0,2\}$ & $\{1\}$ & $\{0\}$ & $\{\}$\\
\end{tabular}
\caption{Sets $\Phi(T)$ of possible values $v=x\xor y$ in $\{0,1,2,3\}$ s.t. $x+y=T$ and $x,y\ge 0$.}
\label{tab:Phi}
\end{table}

We now prove $\nimp{n}$ block by block, following the sequence above. Throughout, $b=(k-4)/8=2^{h-3}-1\ge0$.
 
\begin{lemma}[Block 1]\label{lem:block1}
For $n=1,\ldots,k+1$, $\nimp{n}$ is $0$ or $1$ depending if $n$ is even or odd.
\end{lemma}
 
\begin{proof}
By Thm.~\ref{thm-nimber-monof} as, for $n\le k+1$, the monophonic and $\ell_k$-convexities in $P_n$ are equivalent.
\end{proof}
 
\begin{lemma}[Block 2]\label{lem:block2}
$\nimp{k+1}=1$, $\nimp{k+2}=2$ and $\nimp{k+3}=2$.
\end{lemma}
 
\begin{proof}
From Lemma~\ref{lem:block1}, $\nimp{k+1}=1$.
Recall that $\rho=\frac{k+2}{2}=2^{h-1}-1$. We use $\nimp{n}=\mex\{\nimpb{i}\xor\nimpb{n+1-i}:i=1,\ldots,n\}$ and Lemma~\ref{lem:nimpb-small}.
 
For $n=k+2$: at $i=1$, $\nimpb{1}=0$ and $\nimpb{k+2}=0$ (Lemma~\ref{lem:nimpb-small}), giving the term $0$. For $i=x+1$ with $x=1,\ldots,k/2$: $n+1-i=k+2-x\le k+1$, giving the term $x\xor(k+1-x)$, odd (as $k+1$ is odd). Since $k$ is even, $k/2$ is an integer, and as $k=2^h-4$ is a multiple of $4$, $k/2$ is itself even; at $x=k/2$, $(k/2)\xor(k/2+1)=1$ (consecutive integers with $k/2$ even differ only in the last bit). So $0,1$ are both in the set. Since every other term is either $0$ (at $x=0$) or odd (for $x\ge1$), the value $2$ never occurs, so $\nimp{k+2}=2$.
 
For $n=k+3$: at $i=1,2$, both give the term $1$ (checking directly with Lemma~\ref{lem:nimpb-small}, since $\nimpb{1}=0,\nimpb{k+3}=1$ and $\nimpb{2}=1,\nimpb{k+2}=0$). At $i=x+1$ with $x=k/2+1$ (self-pair, $n+1-i=i$), the term is $0$. At $i=x+1$ with $x=2,\ldots,k/2$, the term is $x\xor(k+2-x)$. This is never $2$: if $x\xor y=2$ with $x+y=k+2$, then by Lemma~\ref{lem:xor-realize}, $x\mathbin{\&}y=(k+2-2)/2=k/2=\rho-1$. Since $\rho=2^{h-1}-1$ has bits $0,\ldots,h-2$ set and $h\ge3$, $\rho-1$ has bit $1$ set; so $x\mathbin{\&}y$ has bit $1$ set, forcing both $x,y$ to have bit $1$ set, so $x\xor y$ has bit $1$ clear, contradicting $x\xor y=2$. Hence $0,1$ are in the set and $2$ is not, so $\nimp{k+3}=2$.
\end{proof}
 
\begin{lemma}[Block 3]\label{lem:block3}
Let $0\le t\le 8b-1=k-5$ and $n=k+4+t$. Then
\[
\nimp{n} = B[t\%8], \qquad B = (1,1,3,3,1,1,2,2).
\]
\end{lemma}
 
\begin{proof}
Since $n\le2k-1$, every $i\in\{1,\ldots,n\}$ and its partner $j=n+1-i$ lie in $\{1,\ldots,2k-1\}$, a range covered by Lemma~\ref{lem:nimpb-small}. So let $g(i)=i-1$ if $i\le k+1$ and $g(i)=i-(k+2)$ if $k+2\le i\le 2k-1$. Then $\nimpb{i}=g(i)$ for every $i\le2k-1$.
Since $t+3\le k-2$, the indices $i=1,\ldots,n$ split into the following three nonempty zones:
 
\begin{itemize}[nosep]
\item[(Z1)] $i=1,\ldots,t+3$: here $i\le k+1$ and $j=n+1-i\ge k+2$;
\item[(Z2)] $i=t+4,\ldots,k+1$: here both $i,j\le k+1$;
\item[(Z3)] $i=k+2,\ldots,n$: mirrors (Z1) under $i\leftrightarrow j=n+1-i$ (by commutativity of $\xor$).
\end{itemize}
Hence $\nimp{n}=\mex(\mathcal V_1\cup\mathcal V_2)$, where, setting $x=i-1$,
\begin{align*}
\mathcal V_1 = &\{x\xor(T_1-x) : x=0,\ldots,T_1\}, &T_1\ =\ &t+2,\\
\mathcal V_2 = &\{x\xor(T_2-x) : x=t+3,\ldots,k\}, &T_2\ =\ &n-1=k+3+t,
\end{align*}
since, for (Z1), $g(i)=i-1=x$ and $g(n+1-i)=n-x-(k+2)=T_1-x$, and, for (Z2), $g(i)=x$ and $g(n+1-i)=n-i=T_2-x$.
 
On (Z1), $\mathcal V_1=\{x\xor(T_1-x):x=0,\ldots,T_1\}$ is exactly the full family of Lemma~\ref{lem:xor-realize} with $T=T_1=t+2$, so $\mathcal V_1\cap\{0,1,2,3\}=\Phi(T_1)$, determined by Table~\ref{tab:Phi} via $T_1\%8=(t+2)\%8$.
 
On (Z2), let $\mathcal V_2^{\mathrm{full}}=\{x\xor(T_2-x):x=0,\ldots,T_2\}\supseteq\mathcal V_2$ (as $\{t+3,\ldots,k\}\subseteq\{0,\ldots,T_2\}$) be the corresponding full family. By Lemma~\ref{lem:xor-realize}, $\mathcal V_2^{\mathrm{full}}\cap\{0,1,2,3\}=\Phi(T_2)$, determined by $T_2\%8$. Since $T_2=k+3+t=(2^h-1)+t$ and $2^h-1\equiv7\pmod8$ (as $h\ge3$), we have $T_2\%8=(t+7)\%8$, which depends only on $t\%8$.
 
Combining these two zones mod $8$: using Table~\ref{tab:Phi} with $T_1\equiv t+2$ and $T_2\equiv t+7\pmod8$, we get $\Phi(T_1)\cup\Phi(T_2)$ in Table~\ref{tab:PhiT1T2} for $t\%8=0,\ldots,7$, satisfying $\mex(\Phi(T_1)\cup\Phi(T_2))=B[t\%8]$.
 
\begin{table}[ht]
\centering
\begin{tabular}{c|c|c|c|c}
$t\%8$ & $\Phi(T_1)$ & $\Phi(T_2)$ & $\Phi(T_1)\cup\Phi(T_2)$ & $\mex(\Phi(T_1)\cup\Phi(T_2))$\\\hline
0 & $\{0,2\}$ & $\{\}$ & $\{0,2\}$ & 1\\
1 & $\{3\}$ & $\{0\}$ & $\{0,3\}$ & 1\\
2 & $\{0,2\}$ & $\{1\}$ & $\{0,1,2\}$ & 3\\
3 & $\{1\}$ & $\{0,2\}$ & $\{0,1,2\}$ & 3\\
4 & $\{0\}$ & $\{3\}$ & $\{0,3\}$ & 1\\
5 & $\{\}$ & $\{0,2\}$ & $\{0,2\}$ & 1\\
6 & $\{0\}$ & $\{1\}$ & $\{0,1\}$ & 2\\
7 & $\{1\}$ & $\{0\}$ & $\{0,1\}$ & 2
\end{tabular}
\caption{Sets $\Phi(T_1)\cup\Phi(T_2)$ for $T_1\equiv t+2$ and $T_2\equiv t+7\pmod8$.}
\label{tab:PhiT1T2}
\end{table}
 
Two things remain. \textbf{(i)} Since $\mathcal V_1\cup\mathcal V_2\subseteq \mathcal V_1\cup\mathcal V_2^{\mathrm{full}}$, the value $B[t\%8]$, being absent from $\Phi(T_1)\cup\Phi(T_2)\supseteq(\mathcal V_1\cup\mathcal V_2)\cap\{0,1,2,3\}$, is also absent from $\mathcal V_1\cup\mathcal V_2$. \textbf{(ii)} Every $v<B[t\%8]$ must be shown present in $\mathcal V_1\cup\mathcal V_2$ itself. Since $B[t\%8]\le3$, only values $v\le2$ are needed. If such a $v$ lies in $\Phi(T_1)$, it is in $\mathcal V_1$; otherwise it lies in $\Phi(T_2)$, and we show that for every $v\in\Phi(T_2)$ with $v\le2$, the witness $w=(T_2-v)/2$ of Lemma~\ref{lem:xor-realize} lies in the range $\{t+3,\ldots,k\}$ of $\mathcal V_2$. Indeed, using $v\le2$ and $t\le 8b-1=k-5=2^h-9$,
\[
w-(t+3) = \frac{(2^h-1+t-v)-2(t+3)}{2} = \frac{(2^h-7-v)-t}{2}\ \ge\ \frac{(2^h-9)-t}{2}\ \ge\ 0,
\]
and, using $v\ge0$,
\[
k-w = (2^h-4)-\frac{2^h-1+t-v}{2} = \frac{2^h-7-t+v}{2}\ \ge\ \frac{2^h-7-(2^h-9)}{2}=1>0.
\]
Hence $t+3\le w\le k-1$, and therefore $v=w\xor(T_2-w)\in\mathcal V_2$. Combined with (i), $\mex(\mathcal V_1\cup\mathcal V_2)=B[t\%8]$ for every $t=0,\ldots,8b-1$, which is precisely $\nimp{n}$ for $n=k+4+t$.
\end{proof}
 
\begin{lemma}[Block 4]\label{lem:block4}
$\nimp{2k}=\nimp{2k+1}=\nimp{2k+2}=1$ and $\nimp{2k+3}=\nimp{2k+4}=2$.
\end{lemma}
 
\begin{proof}
We use $\nimp{n}=\mex\{\nimpb{i}\xor\nimpb{n+1-i}:i=1,\ldots,n\}$ and Lemma~\ref{lem:nimpb-small}: let $g(m)=m-1$ for $m\le k+1$, $g(m)=m-(k+2)$ for $k+2\le m\le2k+3$, and $g(m)=m-(2k+3)$ for $2k+4\le m\le3k+4$, so $\nimpb{m}=g(m)$ for any $m\le3k+4$. For $n=2k,\ldots,2k+3$ every index $i=1,\ldots,n$ and its partner $j=n+1-i$ lie in $\{1,\ldots,2k+3\}$, so only the first two ranges of $g$ are needed; for $n=2k+4$ the single extra index $i=2k+4$ (and its partner $j=1$) touches the 3rd range, $g(2k+4)=1$.
Throughout we split $i=1,\ldots,n$ according to whether $i$ and $j=n+1-i$ fall in the 1st ($\le k+1$) or 2nd range ($\ge k+2$) of $g$, using $k=2^h-4\equiv4\pmod8$ with $h\ge3$ to read off $T\%8$ for each family below, and Table~\ref{tab:Phi} for $\Phi(T)$.
 
For $n=2k$: for $i=k$ (partner $j=k+1$), the term is $(k-1)\xor k=(111)_2=7$, since $k=(1\cdots1\,100)_2$ and $k-1=(1\cdots1\,011)_2$ (with $h-3$ leading ones).
For $1\leq i<k$ (partner $j=2k,\ldots,k+2$), with $x=i-1=0,\ldots,k-2$, the term is $x\xor((k-2)-x)$, the full family of Lemma~\ref{lem:xor-realize} with $T=k-2\equiv2\pmod8$, so $\Phi(T)=\{0,2\}$. Then $\nimp{2k}=\mex(\{0,2,7\}\cap\{0,1,2,3\})=1$.
 
For $n=2k+1$: the self-pair $i=j=k+1$ contributes $0$. For $i=1,\ldots,k$ (partner $j=2k+1,\ldots,k+2$), with $x=i-1=0,\ldots,k-1$, the term is $x\xor((k-1)-x)$, the full family with $T=k-1\equiv3\pmod8$, so $\Phi(T)=\{3\}$. Then $\nimp{2k+1}=\mex\{0,3\}=1$.
 
For $n=2k+2$: for $i=1,\ldots,k+1$ (partner $j=2k+2,\ldots,k+2$), with $x=i-1=0,\ldots,k$, the term is $x\xor(k-x)$, the full family with $T=k\equiv4\pmod8$, so $\Phi(T)=\{0,2\}$. Then $\nimp{2k+2}=\mex\{0,2\}=1$.
 
For $n=2k+3$: the self-pair $i=j=k+2$ gives $0$. For $i=1,\ldots,k+1$ (partner $j=2k+3,\ldots,k+3$), the term is $x\xor((k+1)-x)$ with $x=i-1$, the full family with $T=k+1\equiv5\pmod8$, so $\Phi(T)=\{1\}$. Then $\nimp{2k+3}=\mex\{0,1\}=2$.
 
For $n=2k+4$: the pair $(i=1,j=2k+4)$ gives $g(1)\xor g(2k+4)=0\xor1=1$. The pair $(i=k+2,j=k+3)$ gives $g(k+2)\xor g(k+3)=0\xor1=1$. For $i=2,\ldots,k+1$ (partner $j=2k+3,\ldots,k+4$), with $x=i-1=1,\ldots,k$, the term is $x\xor((k+2)-x)$: this is the family of Lemma~\ref{lem:xor-realize} with $T=k+2\equiv6\pmod8$, restricted to $x\in\{1,\ldots,k\}$ instead of the full $\{0,\ldots,k+2\}$. By Table~\ref{tab:Phi}, the full family meets $\{0,1,2,3\}$ only in $\Phi(T)=\{0\}$, and $0$ does occur, since Lemma~\ref{lem:xor-realize}'s witness $x=w=(T-0)/2=k/2+1$ lies in $\{1,\ldots,k\}$ (as $k\ge4$). Then $\nimp{2k+4}=\mex\{0,1\}=2$.
\end{proof}
 
\begin{lemma}[Block 5]\label{lem:block5}
Let $0\le t\le 8b-1=k-5$ and $n=2k+5+t$. Then
\[
\nimp{n} = B[t\%8], \qquad B = (1,1,3,3,1,1,2,2).
\]
\end{lemma}
 
\begin{proof}
Since $n\le3k$, any $i\in\{1,\ldots,n\}$ and its partner $j=n+1-i$ lie in $\{1,\ldots,3k\}$, so the three ranges of $g$ from Lemma~\ref{lem:nimpb-small} ($g(m)=m-1$ for $m\le k+1$, $g(m)=m-(k+2)$ for $k+2\le m\le2k+3$, $g(m)=m-(2k+3)$ for $2k+4\le m\le3k+4$) suffice.
There are five consecutive non-overlapping zones for $i=1,\ldots,n$:
\begin{itemize}[nosep]
\item[(Z1)] $i=1,\ldots,t+2$ (1st range): here $j=n+1-i\ge2k+4$ (3rd range);
\item[(Z2)] $i=t+3,\ldots,k+1$ (1st range): here $j=n+1-i\in[k+2,2k+3]$ (2nd range);
\item[(Z3)] $i=k+2,\ldots,k+4+t$ (2nd range): here $j=n+1-i\in[k+2,2k+3]$ (2nd range);
\item[(Z4)] $i=k+5+t,\ldots,2k+3$: mirrors (Z2) under $i\leftrightarrow n+1-i$ (by commutativity of $\xor$);
\item[(Z5)] $i=2k+4,\ldots,n$: mirrors (Z1).
\end{itemize}
 
Since (Z4) and (Z5) contribute nothing beyond (Z2) and (Z1), $\nimp{n}=\mex(\mathcal V_1\cup\mathcal V_2\cup\mathcal V_3)$.
Zones Z1 and Z3 combine into a single full family. With $x=i-1$ on (Z1) (so $x=0,\ldots,t+1$): $g(i)=x$ and, since $j=n+1-i=n-x\ge2k+4$, $g(n+1-i)=(n-x)-(2k+3)=(t+2)-x$; the term is $x\xor((t+2)-x)$. With $x=i-(k+2)$ on (Z3) (so $x=0,\ldots,t+2$): $g(i)=x$ and $n+1-i=(k+4+t)-x\in[k+2,2k+3]$, so $g(n+1-i)=\bigl((k+4+t)-x\bigr)-(k+2)=(t+2)-x$; the term is again $x\xor((t+2)-x)$. Thus (Z1) and (Z3) are both sub-ranges of the same family $\{x\xor(T_1-x):x=0,\ldots,T_1\}$ with $T_1=t+2$, and (Z3) alone already covers it in full (as $x$ ranges over all of $0,\ldots,T_1$ there). Hence $\mathcal V_1\cup\mathcal V_3$ equals the full family of Lemma~\ref{lem:xor-realize} with $T=T_1=t+2$, so $(\mathcal V_1\cup\mathcal V_3)\cap\{0,1,2,3\}=\Phi(T_1)$, determined by Table~\ref{tab:Phi} via $T_1\%8=(t+2)\%8$.
 
On Zone Z2, writing $x=i-1$ (so $x=t+2,\ldots,k$): $g(i)=x$, and since $j=n+1-i=n-x\in[k+2,2k+3]$, $g(n+1-i)=(n-x)-(k+2)=(k+3+t)-x$; the term is $x\xor((k+3+t)-x)$, the family of Lemma~\ref{lem:xor-realize} with $T_2=k+3+t$, restricted to $x\in\{t+2,\ldots,k\}$ instead of the full $\{0,\ldots,T_2\}$. Let $\mathcal V_2^{\mathrm{full}}=\{x\xor(T_2-x):x=0,\ldots,T_2\}\supseteq\mathcal V_2$; by Lemma~\ref{lem:xor-realize}, $\mathcal V_2^{\mathrm{full}}\cap\{0,1,2,3\}=\Phi(T_2)$. Since $k\equiv4\pmod8$, $T_2=k+3+t\equiv(t+7)\pmod8$.
 
Combining mod $8$: since $T_1\equiv t+2$ and $T_2\equiv t+7\pmod8$ are exactly as in Table~\ref{tab:PhiT1T2} in the proof of Lemma~\ref{lem:block3} (with $t$ in the same role), $\mex\bigl(\Phi(T_1)\cup\Phi(T_2)\bigr)=B[t\%8]$.
 
As in Lemma~\ref{lem:block3}, since $\mathcal V_1\cup\mathcal V_2\cup\mathcal V_3\subseteq(\mathcal V_1\cup\mathcal V_3)\cup\mathcal V_2^{\mathrm{full}}$, the value $B[t\%8]$ is absent from the actual set; it remains to show every $v<B[t\%8]$ supplied only by $\Phi(T_2)$ (i.e., $v\in\{0,1,2\}$, as $B[t\%8]\le3$) has its witness $x=w=(T_2-v)/2$ in the range $\{t+2,\ldots,k\}$ of $\mathcal V_2$:
\[
w-(t+2) = \frac{(k+3+t-v)-2(t+2)}{2} = \frac{(k-1-t-v)}{2}\ \ge\ \frac{(k-1)-(k-5)-2}{2} = 1 > 0,
\]
using $v\le2$ and $t\le8b-1=k-5$. Hence $w\ge t+3$. Also
\[
k-w = k-\frac{k+3+t-v}{2} = \frac{k-3-t+v}{2}\ \ge\ \frac{k-3-(k-5)+0}{2} = 1 > 0,
\]
using $v\ge0$ and again $t\le k-5$. Hence $w\le k-1$. So $w\in\{t+2,\ldots,k\}$.
Hence every needed value $v<B[t\%8]$ is present in $\mathcal V_1\cup\mathcal V_2\cup\mathcal V_3$, and $B[t\%8]$ is absent, so $\mex(\mathcal V_1\cup\mathcal V_2\cup\mathcal V_3)=B[t\%8]$ for every $t=0,\ldots,8b-1$, which is precisely $\nimp{n}$ for $n=2k+5+t$.
\end{proof}
 
\begin{lemma}[Block 6]\label{lem:block6}
$\nimp{3k+1}=\nimp{3k+2}=1$, $\nimp{3k+3}=3$, $\nimp{3k+4}=0$, and $\nimp{3k+5}=2$.
\end{lemma}
 
\begin{proof}
We use $\nimp{n}=\mex\{\nimpb{i}\xor\nimpb{n+1-i}:i=1,\ldots,n\}$ and the four ranges of $g$ from Lemma~\ref{lem:nimpb-small}: $g(m)=m-1$ for $m\le k+1$; $g(m)=m-(k+2)$ for $k+2\le m\le2k+3$; $g(m)=m-(2k+3)$ for $2k+4\le m\le3k+4$; $g(m)=m-(3k+4)$ for $3k+5\le m\le4k+3$; so $\nimpb{m}=g(m)$ for every $m\le4k+3$. For $n=3k+1,\ldots,3k+4$ every index $i$ and its partner $j$ lie in $\{1,\ldots,3k+4\}$, so only the first three ranges are needed; for $n=3k+5$ the two extra indices $i=1,3k+5$ (mirrors of each other) touch the fourth range, $g(3k+5)=1$. Throughout, $k\equiv4\pmod8$ ($h\ge3$), and for the boundary computations we use the binary representations $k-2=(1\cdots1\,010)_2$, $k-1=(1\cdots1\,011)_2$, $k=(1\cdots1\,100)_2$ and $k+1=(1\cdots1\,101)_2$, each with the same $h-3$ leading ones.
 
For $n=3k+1$: there is no pair $i,j\le k+1$ and no pair $i,j\ge2k+4$. For $i=1,\ldots,k-2$ (partner $j$ in the 3rd range, $x=i-1=0,\ldots,k-3$) and for $i=k+2,\ldots,2k$ (partner in the 2nd range, $x=i-(k+2)=0,\ldots,k-2$), both give the term $x\xor((k-2)-x)$; together they realize the full family of Lemma~\ref{lem:xor-realize} with $T=k-2\equiv2\pmod8$, so it meets $\{0,1,2,3\}$ in $\Phi(T)=\{0,2\}$. The remaining indices $i=k-1,k,k+1$ give the terms $(k-1)\xor k$ and $(k+1)\xor(k-2)$, both equal to $(111)_2=7$, since $(010)_2\xor(101)_2=(011)_2\xor(100)_2=(111)_2$. Then $\nimp{3k+1}=\mex(\{0,2,7\}\cap\{0,1,2,3\})=1$.
 
For $n=3k+2$: for $i=1,\ldots,k-1$ (partner $j$ in the 3rd range, $x=i-1=0,\ldots,k-2$) and $i=k+2,\ldots,2k+1$ (partner $j$ in the 2nd range, $x=i-(k+2)=0,\ldots,k-1$), both give the term $x\xor(k-1-x)$, together realizing the full family with $T=k-1\equiv3\pmod8$, so $\Phi(T)=\{3\}$. The remaining $i=k,k+1$ give $(k-1)\xor(k+1)=(011)_2\xor(101)_2=6$ and $k\xor k=0$. Then $\nimp{3k+2}=\mex(\{3,6,0\}\cap\{0,1,2,3\})=1$.
 
For $n=3k+3$: for $i=1,\ldots,k$ (partner $j$ in the 3rd range, $x=i-1=0,\ldots,k-1$) and $i=k+2,\ldots,2k+2$ (partner $j$ in the 2nd range, $x=i-(k+2)=0,\ldots,k$), both give the term $x\xor(k-x)$, together realizing the full family with $T=k\equiv4\pmod8$, so $\Phi(T)=\{0,2\}$. The remaining index $i=k+1$ gives $k\xor(k+1)=1$. Then $\nimp{3k+3}=\mex\{0,2,1\}=3$.
 
For $n=3k+4$: every index lies in one of $i=1,\ldots,k+1$ (partner $j$ in the 3rd range), $i=k+2,\ldots,2k+3$ (partner $j$ in the 2nd range), or its mirror $i=2k+4,\ldots,3k+4$. Writing $x=i-1=0,\ldots,k$ on the 1st range and $x=i-(k+2)=0,\ldots,k+1$ on the 2nd, both give the term $x\xor((k+1)-x)$: the full family with $T=k+1\equiv5\pmod8$, so $\Phi(T)=\{1\}$. Then $\nimp{3k+4}=\mex\{1\}=0$.
 
For $n=3k+5$: the pair $(i=1,j=3k+5)$ contributes $g(1)\xor g(3k+5)=0\xor1=1$; the pair $(i=k+2,j=2k+4)$ gives $g(k+2)\xor g(2k+4)=0\xor1=1$ (the 1st point of the 2nd range against the 1st point of the 3rd). For $i=2,\ldots,k+1$ (partner $j$ in the 3rd range, $x=i-1=1,\ldots,k$) and $i=k+3,\ldots,2k+3$ (partner in the 2nd range, $x=i-(k+2)=1,\ldots,k+1$), both give the term $x\xor((k+2)-x)$: the family of Lemma~\ref{lem:xor-realize} with $T=k+2\equiv6\pmod8$, restricted to $x\in\{1,\ldots,k+1\}$ instead of the full $\{0,\ldots,k+2\}$. By Table~\ref{tab:Phi}, the full family meets $\{0,1,2,3\}$ only in $\Phi(T)=\{0\}$, so the restricted family meets it in a subset of $\{0\}$; and $0$ does occur, since Lemma~\ref{lem:xor-realize}'s witness $x=w=(T-0)/2=k/2+1$ lies in $\{1,\ldots,k+1\}$ (as $k\ge4$). Then $\nimp{3k+5}=\mex\{0,1\}=2$.
\end{proof}
 
\begin{lemma}[Periodic tail]\label{lem:tail}
Given $p\ge0$ and $0\le t\le k$, let $n=3k+6+p(k+1)+t$, and let $\alpha(t)$ be the symbol at position $t$ (counting from $0$) of the word $B^bC$, that is, $\alpha(t)=B[t\%8]$ if $0\le t\le8b-1$, and $\alpha(t)=C[t-8b]$ if $8b\le t\le k$, where $B=(1,1,3,3,1,1,2,2)$ and $C=(1,1,3,2,3)$. Then $\nimp{n}=\alpha(t)$.
\end{lemma}
 
\begin{proof}
Let $g(m):=\nimpb{m}$ and $X:=\{1,\,k+2,\,4k+4\}$. By Lemma~\ref{lem:nimpb-small}, $g(m)=0$ for $m\in X$ and $g(m)=\sigma(m-1)$ otherwise. We have $\nimp{n}=\mex S$, where $S=\{g(i)\xor g(j):i+j=n+1,\ i,j\geq 1\}$.
Since $n-1=(p+3)(k+1)+t+2$, we have $n-1\equiv t+2\pmod{k+1}$. Let
\[
T_1:=t+2,\qquad T_2:=t+k+3,\qquad R(T):=\bigl\{x\xor(T-x)\ :\ x,\,T-x\in\{1,\ldots,k+1\}\bigr\}.
\]
Since $R(T)$ is contained in the family $\{x\xor(T-x):0\le x\le T\}$, we have $R(T)\cap\{0,1,2,3\}\subseteq\Phi(T)$. Since $k\equiv4\pmod8$, $T_1\equiv t+2$ and $T_2\equiv t+7\pmod8$.
 
\medskip\noindent\textbf{Step 1: $S\subseteq R(T_1)\cup R(T_2)\cup\{0,\ \sigma(t+2),\ \sigma(t+3)\}$.} Consider a pair $i+j=n+1$.
\begin{itemize}[nosep]
\item If $i,j\notin X$: the term is $x\xor y$ with $x=\sigma(i-1)$, $y=\sigma(j-1)\in\{1,\ldots,k+1\}$, where $x+y\equiv i+j-2=n-1\equiv T_1\pmod{k+1}$. Since $2\le x+y\le2k+2$ and $2\le T_1\le k+2$, we get $x+y\in\{T_1,T_2\}$, so the term lies in $R(T_1)\cup R(T_2)$.
\item If $i,j\in X$: the term is $0$.
\item If $i\in X$ and $j\notin X$: if $i\in\{1,k+2\}$, then $j-1\in\{n-1,\,n-k-2\}$, both congruent to $n-1\equiv t+2$, so the term is $\sigma(t+2)$. If $i=4k+4$, then $j-1=n-1-(4k+3)\equiv t+3$, because $4k+3\equiv-1\pmod{k+1}$, so the term is $\sigma(t+3)$.
\end{itemize}
 
\medskip\noindent\textbf{Step 2: $R(T_1)\cup R(T_2)\subseteq S$.} Let $x,y\in\{1,\ldots,k+1\}$ with $x+y\in\{T_1,T_2\}$. Take $i=k+2+x\in\{k+3,\ldots,2k+3\}$ and $j=n+1-i\ge n-2k-2\ge k+4$. Then $i\notin X$ and $\sigma(i-1)=x$. Also $j\notin\{1,k+2\}$ and $j-1\equiv n-1-x\equiv T_1-x\equiv y\pmod{k+1}$, so $\sigma(j-1)=y$. If $j\ne4k+4$, the pair $\{i,j\}$ gives $x\xor y$. If $j=4k+4$, use instead $i'=i+(k+1)\in\{2k+4,\ldots,3k+4\}$ and $j'=j-(k+1)=3k+3$. Neither is in $X$, they have the same residues as $i$ and $j$, and $i'+j'=n+1$, so this pair also gives $x\xor y$.
 
\medskip\noindent\textbf{Step 3: witnesses.} Let $v\in\Phi(T)$. By Lemma~\ref{lem:xor-realize}, $x=(T-v)/2$ is a witness, with partner $T-x=(T+v)/2\ge x$. Hence $v\in R(T)$ as soon as $(T-v)/2\ge1$ and $(T+v)/2\le k+1$. For $T=T_1$ and $v\le3$, the second condition always holds (as $t+v\le k+3\le2k$), and the first reads $v\le t$. For $T=T_2$, the first condition always holds, and the second reads $t\le k-1-v$.
 
\medskip
By Steps 1 and 2, $\nimp{n}=\alpha(t)$ follows once we show that (a) $\alpha(t)\notin R(T_1)\cup R(T_2)\cup\{0,\sigma(t+2),\sigma(t+3)\}$, and (b) every $v<\alpha(t)$ lies in $S$.
 
\medskip\noindent\textbf{Case $0\le t\le k-2$.} Since $8b=k-4\equiv0\pmod8$ and $C[0],C[1],C[2]=1,1,3=B[0],B[1],B[2]$, we have $\alpha(t)=B[t\%8]$ throughout this case. Here $\sigma(t+2)=t+2$ and $\sigma(t+3)=t+3$; both exceed $3$ when $t\ge2$, and when $t\le1$ we have $\alpha(t)=B[t]=1\notin\{0,t+2,t+3\}$. Since $T_1\equiv t+2$ and $T_2\equiv t+7\pmod8$ are the residues of Table~\ref{tab:PhiT1T2} in the proof of Lemma~\ref{lem:block3}, we get $\alpha(t)\notin\Phi(T_1)\cup\Phi(T_2)\supseteq(R(T_1)\cup R(T_2))\cap\{0,1,2,3\}$, which proves (a), and every $v<\alpha(t)$ (so $v\le2$) lies in $\Phi(T_1)\cup\Phi(T_2)$. For (b), we apply Step~3. If $v\in\Phi(T_1)$, we need $v\le t$: this holds for $t\ge2$ because $v\le2$, and for $t\le1$ because then $\alpha(t)=1$ and only $v=0$ is needed. Otherwise $v\in\Phi(T_2)$, and we need $t\le k-1-v$: this holds for $t\le k-3$ because $v\le2$, and for $t=k-2$ (where $t\%8=2$ and $\Phi(T_2)=\{1\}$) because $v=1$.
 
\medskip\noindent\textbf{Case $t=k-1$} (so $\alpha(t)=C[3]=2$). Here $T_1=k+1\equiv5\pmod8$, so $\Phi(T_1)=\{1\}$, and $T_2=2k+2$; since $x+y=2k+2$ with $x,y\le k+1$ forces $x=y=k+1$, we get $R(T_2)=\{0\}$. The extra values are $\sigma(k+1)=k+1$ and $\sigma(k+2)=1$. Hence $2\notin S$, while $0\in R(T_2)$ and $1\in R(T_1)$ (as $1\le t$), so $\nimp{n}=2$.
 
\medskip\noindent\textbf{Case $t=k$} (so $\alpha(t)=C[4]=3$). Here $T_1=k+2\equiv6\pmod8$, so $\Phi(T_1)=\{0\}$, and $T_2=2k+3>2k+2$, so $R(T_2)=\emptyset$. The extra values are $\sigma(k+2)=1$ and $\sigma(k+3)=2$. Hence $3\notin S$. Moreover, $0\in R(T_1)$ (as $0\le t$); the pair $\{1,n\}$ gives $g(n)=\sigma(n-1)=\sigma(k+2)=1$, since $n=4k+6+p(k+1)\notin X$; and the pair $\{4k+4,\,j\}$ with $j=n-4k-3=p(k+1)+3$ gives $\sigma(j-1)=2$, since $j\equiv3\pmod{k+1}$ implies $j\notin X$ (whose elements are $\equiv1,1,0$). Hence $\nimp{n}=3$.
\end{proof}
 
\begin{theorem}
\label{thm:Pn}
Let $k=2^h-4$ with $h\ge3$. Then the nimber sequence $\nimp{n}$ is periodic with period $k+1$, the transient has length $3k+5$, the only value of $n$ for which $\nimp{n}=0$ with $n>k$ is $n=3k+4$, and the sequence has the following expression:
\begin{align*}
\nimp{n}:&\ \ \ \ \ \ [1,0]^{k/2}, 1, 2, 2, B^b, 1, 1, 1, 2, 2, B^b, 1, 1, 3, 0, 2, [B^b, C]^*,
\end{align*}
where $b=(k-4)/8$, $B=11331122$ and $C=11323$.
\end{theorem}
 
\begin{proof}
Lemmas~\ref{lem:block1},~\ref{lem:block2},~\ref{lem:block3},~\ref{lem:block4},~\ref{lem:block5} and~\ref{lem:block6} establish the theorem for $n=1,\ldots,3k+5$, i.e., for all six blocks $[1,0]^{k/2},\,1,2,2,\,B^b,\,1,1,1,2,2,\,B^b,\,1,1,3,0,2$ of the sequence. The remaining periodic tail $[B^b,C]^*$, for every $n\ge3k+6$, is exactly Lemma~\ref{lem:tail}.
\end{proof}
 
Written as a word, the periodic block of Theorem~\ref{thm:Pn} (of length $k+1$) is $B^{\,b}C$, where $B=11331122$ and $C=11323$. Since $k=2^h-4$, $k-4=8\bigl(2^{h-3}-1\bigr)$, so $b=2^{h-3}-1$ is a nonnegative integer for every $h\ge3$: the period-$8$ block $B$ always divides the length-$(k+1)$ periodic block evenly, leaving exactly the fixed length-$5$ remainder $C$. Notice that, in contrast to $\nimpc{n}$ and $\nimpb{n}$, whose values grow linearly with $k$, $\nimp{n}$ never exceeds $3$ once periodic, and the value $0$ never recurs after the transient.

\section{Concluding remarks and future work}
 
In this paper we initiated the study of hull games in the monophonic and $\ell_k$-convexities.
We proved PSPACE-completeness of $\CHG_{\mc}$ and $\CHG_{\ell_k}$ for every $k\ge2$ even in graphs of diameter at most three, and we gave polynomial time algorithms, via the Sprague-Grundy Theory, to determine the winner of these games in disjoint unions of paths and cycles.
For $k$ odd, we obtained exact closed-form winner criteria in $P_n$ and $C_n$; for the family $k=2^h-4$ with $h\ge3$, we established the full periodic structure of the nimber sequences, again in both $P_n$ and $C_n$. For $k=2$, we uncovered a close relationship with Conway's \emph{Couples-are-Forever}, matching the sequences of $P'_n$ and $P''_n$ exactly and the sequence of $P_n$ itself with agreement rate $99.97\%$ for $n$ up to $200$ million.
 
Several questions remain open.
 
\begin{itemize}
\item The complexity of the interval game $\CIG_\mathcal{C}$ remains open for every convexity $\mathcal{C}$ considered here (geodesic, monophonic and $\ell_k$), in contrast with the hull game, which we now know to be PSPACE-complete in all these cases.
 
\item Our computational experiments (Section~\ref{sec:forever}) strongly suggest, but do not settle, that the nimber sequence of $\nim_{\ell_2}(P_n)$ eventually behaves like a shift of the sequence of \emph{Couples-are-Forever}, up to a small, structured set of exceptions. Since the periodicity of \emph{Couples-are-Forever} itself is a long-standing open problem (Caines et al.~\cite{caines99} reached $n=50$M and we extended this to $200$M), we leave as an open question whether $\nim_{\ell_2}(P_n)$ is eventually periodic, and whether the exceptional positions (the 12 values of $n$ where Alice loses) can be characterized without exhaustive search.
 
\item We showed that, for every $k=2^h-4$ with $h\ge3$, the nimber sequences of $\CHG_{\ell_k}$ in $P_n$ and $C_n$ are eventually periodic, and that this is the only family of even $k$ for which we observed periodicity in extensive computational testing. This suggests the following question: Is that true that the nimber sequence $\nim_{\ell_k}(P_n)$ is aperiodic for every even $k\ge2$ with $k\ne2^h-4$ for any integer $h\ge3$?
 
\item It would be interesting to determine the smallest diameter for which $\CHG_{\mc}$ and $\CHG_{\ell_k}$ (with $k\ge3$) remain PSPACE-hard. Our reduction (Theorem~\ref{teo-pspace1}) achieves diameter three; for $\ell_2$, PSPACE-hardness is already known to hold at diameter two~\cite{araujo24}, so it is natural to ask whether the same holds for $\mc$ and for $\ell_k$ with $k\ge3$.
 
\item Finally, extending the polynomial time algorithms of Section~\ref{sec:poly} beyond disjoint unions of paths and cycles -- e.g., to trees, or to graphs of bounded treewidth -- is a natural next step, given that such extensions are known for several other graph convexity games.
\end{itemize}
 
\section*{Acknowledgments}

During the preparation of this work the authors used Claude (Anthropic) to help write and test the C and Python programs used in the computational experiments on the nimber sequence of Couples-are-Forever reported in Section~\ref{sec:forever} (computation up to $n=200$ million, statistical analyses and figures), to prepare the public data repository associated with those experiments, to check the correctness of the proofs and run numerical verifications of the nimber sequences, to improve the wording of some sentences and the \LaTeX{} formatting of the manuscript, and to check the manuscript for numerical inconsistencies. The mathematical results, the design of the experiments and the interpretation of the results are the authors' own, and every proof was checked by the authors. After using this tool, the authors reviewed and edited the content as needed and take full responsibility for the content of the published article.

\bibliographystyle{plain}
\bibliography{refs}

@misc{rudini26cf,
  author       = {Rudini Sampaio},
  title        = {Nimbers of the game {Couples-are-Forever} up to 100 million},
  howpublished = {\url{https://www.lia.ufc.br/~rudini/publ/cf/couples_are_forever.htm}},
  year         = {2026},
}

@InProceedings{benevides26,
  author =	{Ara\'{u}jo, Samuel N. and Benevides, Fabr{\'\i}cio and Martins, Nicolas and Nisse, Nicolas and Sampaio, Rudini},
  title =	{{The Closed Hull Game and the Closed Interval Game}},
  booktitle =	{13th International Conference on Fun with Algorithms (FUN 2026)},
  pages =	{4:1--4:12},
  series =	{Leibniz International Proceedings in Informatics (LIPIcs)},
  ISBN =	{978-3-95977-417-8},
  ISSN =	{1868-8969},
  year =	{2026},
  volume =	{366},
  doi =		{10.4230/LIPIcs.FUN.2026.4},
}

@article{caines99,
 author = {Ian Caines and Carrie Gates and Richard K. Guy and Richard J. Nowakowski},
 journal = {The American Mathematical Monthly},
 number = {4},
 pages = {359--361},
 title = {Periods in Taking and Splitting Games},
 volume = {106},
 year = {1999}
}

@Article{sprague36,
 Author = {Sprague, R.},
 Title = {{\"U}ber mathematische {Kampfspiele}},
 FJournal = {Tohoku Math J},
 Journal = {T\^ohoku Math J},
 Volume = {41},
 Pages = {438--444},
 Year = {1936},
 Language = {German},
 zbMATH = {3020856},
 Zbl = {0013.29004}
}

@Article{grundy39,
 Author = {Grundy, P. M.},
 Title = {Mathematics and games},
 Journal = {Eureka (The Archimedeans' Journal)},
 Volume = {2},
 Pages = {6--8},
 Year = {1939},
}

@ARTICLE{harary81,
    title = {Convexity in graphs},
    journal = {Journal of Differential Geometry},
    volume = {16},
    number = {1},
    pages = {185--190},
    year = {1981},
    author = {Harary, Frank and Nieminem, Juhani},
    mrnumber = {0638785},
    zbl = {0493.05037},
    doi = {},
    langid = {en},
}

@INCOLLECTION{harary84,
    author = {Harary, Frank},
    title = {Convexity in Graphs: Achievement and Avoidance Games},
    booktitle = {Convexity and Graph Theory},
    publisher = {North-Holland},
    volume = {87},
    pages = {323},
    year = {1984},
    mrnumber = {0791006},
    langid = {en},
    zbl = {},
}

@article{harary84b,
    author = {Harary, Frank},
    title = {Game theoretic aspects of graph theory},
    journal = {Annals of Discrete Mathematics},
    volume = {20},
    pages = {133-135},
    year = {1984},
}

@ARTICLE{buckley85,
    author = {Buckley, Fred and Harary, Frank},
    journal = {Quaestiones Mathematicae},
    pages = {321--334},
    title = {Geodetic games for graphs},
    volume = {8},
    year = {1985},
    mrnumber = {0854054},
    zbl = {0615.90093},
    langid = {en},
    zbl = {0615.90093},
}

@ARTICLE{haynes-2003,
    author = {Haynes, T. W. and Henning, M. A. and Tiller, C.},
    journal = {Quaestiones Mathematicae},
    pages = {389--397},
    title = {Geodetic achievement and avoidance games for graphs},
    volume = {26},
    year = {2003},
    mrnumber = {2046144},
    langid = {en},
    zbl = {1152.05377},
}

@BOOK{demaine-2009,
    author = {Hearn, R. and Demaine, E.},
    publisher = {A. K. Peters Ltd},
    title = {Games, Puzzles and Computation},
    year = {2009},
    mrnumber = {2537584},
    zbl = {1175.91035},
    doi = {},
    langid = {en},
    pages = {ix+237},
    zbl = {1175.91035},
}

@ARTICLE{nec-1988,
    author = {Nečásková, M.},
    journal = {Quaestiones Mathematicae},
    pages = {115--119},
    title = {A note on the achievement geodetic games},
    volume = {12},
    year = {1988},
    mrnumber = {0979252},
    langid = {en},
    zbl = {},
}

@ARTICLE{schaefer-1978,
    title = {On the complexity of some two-person perfect-information games},
    journal = {J. Comput. Syst. Sci.},
    volume = {16},
    number = {2},
    pages = {185--225},
    year = {1978},
    author = {Schaefer, Thomas J.},
    mrnumber = {0490917},
    langid = {en},
    zbl = {0383.90112},
}

@INPROCEEDINGS{zermelo13,
    author = {Zermelo, Ernst},
    booktitle = {Proceedings of the Fifth International Congress of Mathematicians},
    pages = {501--504},
    title = {Über eine {A}nwendung der {M}engenlehre auf die {T}heorie des {S}chachspiels},
    year = {1913},
    mrnumber = {},
    zbl = {},
    langid = {de},
    zbl = {},
}

@ARTICLE{dragan99,
    author = {Dragan, Feodor F. and Nicolai, Falk and Brandstädt, Andreas},
    title = {Convexity and {HHD}-free graphs},
    journal = {SIAM J. Discrete Math.},
    volume = {12},
    year = {1999},
    number = {1},
    pages = {119--135},
    doi = {10.1137/S0895480195321718},
    mrnumber = {1652261},
    zbl = {0916.05060},
    langid = {en},
    zbl = {0916.05060},
}

@article{chandran24,
title = {The general position avoidance game and hardness of general position games},
journal = {Theor. Comput. Sci.},
volume = {988},
pages = {114370},
year = {2024},
author = {Ullas {Chandran S. V.} and Sandi Klavžar and Neethu {P. K.} and Rudini M. Sampaio},
langid = {en},
}

@ARTICLE{jamison82,
    title = {A perspective on abstract convexity: classifying alignments by varieties},
    author = {Jamison, Robert E.},
    journal = {Convexity and Related Combinatorial Geometry},
    year = {1982},
    mrnumber = {0650310},
    zbl = {0482.52001},
    doi = {},
    langid = {en},
    zbl = {0482.52001},
}

@ARTICLE{gutierrez-protti-tondato2023,
    author = {Gutierrez, Marisa and Protti, Fábio and Tondato, Silvia},
    title = {Convex geometries over induced paths with bounded length},
    journal = {Discrete Mathematics},
    year = {2023},
    volume = {346},
    number = {1},
    pages = {113133},
}

@book{araujo-livro25,
title = {{I}ntroduction to {G}raph {C}onvexity: an algorithmic approach},
author = {Araújo, Júlio and Dourado, Mitre C and Protti, Fábio and Sampaio, Rudini M.},
year = {2025},
publisher = {Springer Cham},
}

@article{dailly26,
title = {The closed geodetic game: Algorithms and strategies},
journal = {Theoretical Computer Science},
volume = {1085},
pages = {116186},
year = {2026},
author = {Antoine Dailly and Harmender Gahlawat and Zin Mar Myint},
}

@article{araujo-cocoon25,
author="Araújo, Samuel N. and Brito, J. Marcos
and Folz, Raquel and de Freitas, Rosiane and Sampaio, Rudini M.",
title = {Algorithms and complexity of graph convexity partizan games},
journal = {Theor. Comput. Sci.},
volume = {1044},
number = {},
pages = {115267},
year = {2025},
}

@article{araujo24,
title = {Graph convexity impartial games: Complexity and winning strategies},
journal = {Theor. Comput. Sci.},
volume = {998},
pages = {114534},
year = {2024},
author="Araújo, Samuel N. and Brito, J. Marcos
and Folz, Raquel and de Freitas, Rosiane and Sampaio, Rudini M.",
}
 
\end{document}